\documentclass[12pt, titlepage]{article}
\usepackage[margin=1in]{geometry} 
\usepackage{amssymb} 
\usepackage{amsmath} 
\usepackage{amsthm} 
\usepackage{graphicx} 
\usepackage{booktabs} 
\usepackage{natbib} 
\usepackage{pgfplots} 
\usepgfplotslibrary{fillbetween} 
\usetikzlibrary{decorations.pathreplacing,calligraphy}
\usepackage{bm} 
\usepackage{xcolor} 
\definecolor{royalblue}{HTML}{0000CD} 
\usepackage[colorlinks=true]{hyperref} 
\hypersetup{
    colorlinks=true,
    allcolors=royalblue,
} 
\usepackage{threeparttable} 
\usepackage{thmtools} 
\usepackage{setspace} 
\usepackage{rotating} 

\usepackage{multirow}

\declaretheorem[numberwithin=section]{theorem} 
\newtheorem{corollary}{Corollary}[theorem] 
\allowdisplaybreaks 
\usepackage{mathtools} 

\newtheorem{Lemma}{{Lemma}}
\newtheorem*{remark}{Remark}

\providecommand{\keywords}[1]
{
  \small	
  \textbf{\textit{Keywords---}} #1
}

\usepackage{caption} 
\usepackage{subcaption} 
\usepackage{enumerate} 
\usepackage{enumitem} 

\usepackage{algorithm} 
\usepackage{algpseudocode} 

\makeatletter
\DeclareRobustCommand{\iscircle}{\mathord{\mathpalette\is@circle\relax}}
\newcommand\is@circle[2]{%
  \begingroup
  \sbox\z@{\raisebox{\depth}{$\m@th#1\bigcirc$}}%
  \sbox\tw@{$#1\square$}%
  \resizebox{!}{\ht\tw@}{\usebox{\z@}}%
  \endgroup
}
\makeatother

\begin{document}

\title{Pricing time-to-event contingent cash flows: A discrete-time
survival analysis approach
  \thanks{This material is
based upon work supported by the National Science Foundation
Graduate Research Fellowship under Grant No. DHE 1747453.
Declarations of interest: none.}}

\author{
  Jackson P. Lautier\footnote{Department of Statistics, University of
  Connecticut}%
  \thanks{Corresponding to jackson.lautier@uconn.edu.}
  \and
  Vladimir Pozdnyakov\footnotemark[2]
  \and
  Jun Yan\footnotemark[2]
}

\date{\today}

\maketitle

\doublespacing

\begin{abstract}
Prudent management of insurance investment portfolios requires competent asset
pricing of fixed-income assets with time-to-event contingent cash flows, such
as consumer asset-backed securities (ABS).  Current market pricing techniques
for these assets either rely on a non-random time-to-event model or may not
utilize detailed asset-level data that is now available with most public
transactions.  We first establish a framework
capable of yielding estimates of the time-to-event random variable from
securitization data, which is discrete and often subject to left-truncation
and right-censoring.  We then show that the vector of discrete-time hazard rate
estimators is asymptotically multivariate normal with independent components,
which has not yet been done in the statistical literature in the case of both
left-truncation and right-censoring.  The time-to-event distribution
estimates are then fed into our cash flow model, which is capable of
calculating a formulaic price of a pool of time-to-event contingent cash flows
vis-\'{a}-vis calculating an expected present value with respect to the
estimated time-to-event distribution.  In an application to
a subset of 29{,}845 36-month leases from the Mercedes-Benz Auto Lease Trust
2017-A (MBALT 2017-A) bond, our pricing model yields estimates closer to the
actual realized future cash flows than the non-random time-to-event model,
especially as the fitting window increases. Finally, in certain settings,
the asymptotic properties of the hazard rate estimators allow investors
to assess the potential uncertainty of the price point estimates, which we
illustrate for a subset of 493 24-month leases from MBALT 2017-A.


\end{abstract}

\section{Introduction}

Life insurers hold approximately \$670-\$770 billion in securitized assets
\citep[]{mcmenamin_2013,naic_2020}, which is nearly 16--20\% of all insurer
general account assets.  Of these securitized assets, over \$170 billion are
in asset-backed securities (ABS), or just over 5\% of all general account
holdings.  Proper asset-liability management (ALM) and general asset
management for insurers require pricing cash flows from ABS and related assets.
The actuarial literature leaves ABS largely untouched, however, though there 
are numerous related contributions within general asset-liability management 
(ALM)
\citep[e.g.,][]{yao_2013, chiu_2014, zhang_2016, wei_2017, zhang_2017, li_2018, nolsoe_2020}
and credit risk
\citep[e.g.,][]{liang_2012, gatzert_2012, denuit_2015, guo_2017,
kiatsupaibul_2017, jang_2018}.

Outside an insurance context, the
literature for valuing an asset-backed security may be loosely
categorized into two alternative approaches: modeling at the pool-level or a
top-down approach and modeling at the individualized loan-level or a
bottom-up approach.  For a general introduction to each, see
\citet[][Chapters 7, 12]{davidson_2014}.
For a specific example of a top-down approach that connects portfolio level
losses to interest rates, see \citet{fermanian_2013}.  Alternatively, investors
may rely solely on input from credit rating agencies or utilize credit ratings
in combination with commercial cash flow models that do not incorporate the
potential randomness of the individual consumer credits within the trust. Both
of these approaches are considered to be inadequate when compared with a cash
flow model that incorporates the stochastic nature of underlying credit risks,
however \citep[][Chapter 8]{bluhm_2010}.  Finally, investors may rely on the
prescribed schedule provided by an ABS prospectus, which assumes underlying
cash flows occur as scheduled \citep{mercedes_2017}.  We find this
approach may be inadequate to capture true trust performance in our comparative
analysis (Section~\ref{sec:app}), however.

Recent improvements in data availability should also be
considered.  For example, with the enactment of U.S. Securities and Exchange
Commission (SEC)
Regulation AB II in November 2016 \citep{reg_ab2}, which requires issuers of
publicly traded securities to disclose pertinent asset-level demographic and
performance data, investors now have the ability to model most forms of ABS at
the loan level via a bottom-up
approach.  Indeed, most prospective ABS investors download this asset-level
information during the pricing period of a newly issued ABS bond
\citep{neilson_2022}.  Thus, to be consistent with industry best
practices, we will utilize the asset-level data to estimate the probabilistic
distributions underlying a stochastic loan-level present value cash flow model.
Its expected value is then the point estimate of the price of a single asset,
and the sum total of the individual expected value calculations for all active
time-to-event contingent assets is then the price of the complete ABS.

Specifically, our model is capable of incorporating four sources of randomness:
(1) the random time-until-contract-termination,
(2) the random number of months past the last monthly payment until the
residual is paid into the trust,
(3) the random residual value realization amount when a leased automobile is
sold to repay the trust, and, optionally,
(4) the estimator uncertainty for the time-until-contract-termination
distribution.
Of these four sources of randomness, it is critically important to achieve an
accurate estimation of the time-to-event probability distribution.

A close examination
of the estimation problem for the time-to-event random variable from ABS
investment trust data reveals that one must account for incomplete data in the
form of left-truncation and right-censoring.  There is a long history of
calculating point estimates under these circumstances \citep{tsai_1987}.
Missing in the statistical literature, however, is the asymptotic properties of
such estimates in the case of a discrete-time-to-event distribution.  Because
our data is financial and thus updates only monthly, the asymptotic results
of papers such as \citet{woodroofe_1985} and \citet{tsai_1987}, which assume a
continuous distribution function for the time-to-event random variable,
do not fully address our problem.  Further, inappropriately assuming continuous
time
is problematic because it requires assuming that two events cannot have an
identical termination time (i.e., ties have zero
probability). In a securitization pool of tens of thousands of leases with
identical contract lengths, two leases sharing the same termination age is not
only possible but an almost certainty.
Other approaches to avoid a discrete-time assumption, such as assuming
interval censoring or grouped survival data, also do not adequately address
ABS data because a payment made any time before the due date is treated the
same as a payment on the due date.  In other words, ABS data is truly discrete;
it is not just a result of measurement imprecision.


This paper thus has two contributions to the actuarial literature.
The first is statistical and relates to the establishment of a precise
discrete-time framework for the underlying random variables in an
ABS setting for both random left-truncation and right-censoring
(Section~\ref{sec:prelim}) and derivation of the point estimator in this
setting including its asymptotic properties (Section~\ref{sec:est}).
The second is our proposed formulaic
pricing model that utilizes the discrete-time lease lifetime estimator of
Section~\ref{sec:est} (in conjunction with the aforementioned other
sources of randomness)
to estimate the price of an auto-lease asset-backed security
(Section~\ref{sec:cash_flow}), which outperforms the standard prospectus
approach of \citet{mercedes_2017} (Table~\ref{tab:apv_results}).
While our application focuses on pricing the cash
flows of an auto-lease ABS loan pool, the model generalizes to other
forms of ABS, such as agency mortgage-backed securities (MBS).  To help readers
understand potential misapplications of our model, we provide a
detailed discussion of important assumptions and appropriate use cases 
in Section~\ref{subsec:understand}.  The remaining sections include a
simulation study focusing on the statistical results
(Section~\ref{sec:sim_study}), a numerical application to a subset of
29{,}845 leases from the MBALT 2017-A bond (Section~\ref{sec:app}), and
concluding remarks (Section~\ref{sec:conclusion}).
All proofs of major results and additional Section~\ref{sec:app} details may
be found in the Appendices~\ref{sec:proofs} and \ref{sec:app_detail}, 
respectively.

\section{Preliminaries}
\label{sec:prelim}

We first outline the mathematical details behind attempting
to make meaningful inference about the distribution of a discrete-time lifetime
random variable of interest from left-truncated data (a well-accepted yet nontrivial
claim on close examination).  For those interested, an expanded exposition of
the discrete-time incomplete data case of left-truncation may be found in
\citet{lautier_2021}.  For narrative convenience and given our intended
application, we will work towards defining the mathematical details within the
context of an automotive lease securitization. Next, we generalize the work of
\citet{lautier_2021} to also handle right-censored data because our subsequent
goal is to develop a model capable of pricing an actively paying ABS bond.
(Since the bond is active, there will be leases known to still be paying as of
the pricing time but with a yet unknown termination time.)
The section concludes by detailing the assumptions of our sampling procedure,
in which our data is assumed to be sampled from an already left-truncated
population --- as is the case for ABS investment trust data --- in comparison
to the unsuitable for our application ``truncate after sampling" procedure used
in \citet{woodroofe_1985} and \citet{tsai_1987}.
The rigor of this section is motivated by the continuous-time
analogs of \citet{woodroofe_1985} (left-truncation) and \citet{tsai_1987}
(left-truncation and right-censoring).


\subsection{Left-truncation}

We now begin with the details.
Let $X$ and $Y$ be two independent, positive, and integer-valued discrete
random variables, with distribution functions $F$ and $G$, respectively.
Further assume that we only observe the pairs $(X,Y)$ for which $Y \leq X$.
That is, our observed data is conditional. Hence, let $H_*$ denote the joint
distribution function of $X$ and $Y$ given $Y \leq X$, and let $F_*$ and
$G_*$ denote the marginal distribution functions of
$X$ and $Y$, respectively, given $Y \leq X$. Formally,
\begin{equation}
H_*(F,G,x,y) = \Pr(X \leq x, Y \leq y \mid X \geq Y),
\label{eq:H}
\end{equation}
is the joint conditional distribution with conditional marginal distributions
$F_*$ and $G_*$.   We include $F$ and $G$ within the notation of
$H_*$ to emphasize that equivalent $H_*$ may be constructed from different
$F$ and $G$, a technical point we now clarify.

Let the support of $F$ be $(a_F \leq x \leq b_F)$, where
$0 \leq a_F \leq b_F$ and $a_F, b_F \in \mathbb{Z}$, and let the support of $G$
be $(a_G \leq y \leq b_G)$, where $0 \leq a_G \leq b_G$ and
$a_G, b_G \in \mathbb{Z}$.
There will be complete left-truncation (full data loss) if
$a_G \geq b_F$.  Now, $H_*$ may be constructed from any pairs of $F$ and $G$
such that $(F,G) \in \mathcal{K}$, where
\begin{equation*}
\mathcal{K} = \{(F,G): F(0) = 0 = G(0), \quad \Pr(Y \leq X) > 0\}.
\end{equation*}
Difficulties may arise in the recovery of $(F,G)$ from $H_*$ because there
might exist a different pair $(F_0,G_0)$ that can generate the same $H_*$. 
That is, we have a possible \textit{identifiability} issue.

More specifically, consider the following subset of $\mathcal{K}$:
\begin{equation*}
\mathcal{K}_0 = \{(F,G) \in \mathcal{K} : a_G \leq a_F, \quad
b_G \leq b_F\}.
\end{equation*}
For any $(F,G) \in \mathcal{K}$, but $(F,G) \notin \mathcal{K}_0$, let
$F_0 = \Pr(X \leq x \mid X \geq a_G)$ and
$G_0 = \Pr(Y \leq y \mid Y \leq b_F)$.  Then $(F_0, G_0) \in \mathcal{K}_0$,
and Lemma 1 of \citet{woodroofe_1985} demonstrates $H_*(F_0, G_0) = H_*(F,G)$
for any $(F,G) \in \mathcal{K}$.  In other words, we have two
potential pairs $(F_0, G_0)$ and $(F,G)$ that lead to the same $H_*$. How,
then, can we make inference on $X$ from left-truncated data?

In most applications, we cannot.  But not all is lost.  Indeed, Theorem 1 of
\citet{woodroofe_1985} states that we can find a unique $(F_0, G_0)$
if we restrict our construction of $H_*$ to just the members of
$\mathcal{K}_0$. More formally, for every $H_*$ based on some
$(F,G) \in \mathcal{K}$, there is only one pair $(F_0, G_0) \in \mathcal{K}_0$
such that $H_*(F_0, G_0) = H_*(F,G)$, and this pair is given by $F_0$ and
$G_0$. Theorem 1 of \cite{woodroofe_1985} also shows how to recover the
cumulative hazard functions of $F_0$ and $G_0$ and therefore recover
$F_0$ and $G_0$.  Note that \citet{woodroofe_1985} assumes $F$ and $G$ are
right-continuous in his Lemma 1 and Theorem 1, and so we avoid any
discrete-time complications in using these particular results directly.

Let us now turn to our application. Let $T$ denote the
random time of a new lease contract origination.  We assume $T$ is discrete and
spans the finite range $1 \leq T \leq m$.  A realization of $T$, say $t$,
is then the initial point of the time-until-lease-contract-termination
random variable, our lifetime variable of interest, denoted by $X$.
Lease contracts have
a fixed duration, and we denote this final possible termination time to be
$\omega$, where $\omega \in \mathbb{N}$ and is finite. Since issuers of
structured debt typically have a legal obligation to the trust to select lease
contracts with a minimum history of on-time payments, the youngest least in the
trust will have some minimum age, $\Delta$, where $1 \leq \Delta \leq \omega$.

Thus, the trust begins at time $m + \Delta$, where $\Delta$ is non-random.
If we denote $Y = m + \Delta + 1 - T$, then
$\Delta + 1 \leq Y \leq m + \Delta$. Further, $Y$ represents the minimum
amount of time a lease must remain active to be observed in the trust.  Hence,
we will only observe $X$ given $X \geq Y$, and therefore $Y$ is a left-truncation
random variable.  Additionally, if we assume the time of a new lease contract
origination, $T$, is independent of the time of lease termination, $X$, then
$X$ and $Y$ will be independent.  The assumed independence
of $X$ and $T$ (and therefore $Y$) is vital to this analysis, and it may not
hold in all applications.  For additional details on the appropriateness of
this important assumption within our application, please see
Section~\ref{subsec:understand}.  For completeness, 
$\Delta + 1 \leq X \leq \omega$.

In terms of recovery, therefore, we have $a_G = \Delta + 1$,
$b_G = m + \Delta$, $a_F = \Delta + 1$, and $b_F = \omega$.  Hence, if
$\Delta > 0$,
$F_0 = \Pr(X \leq x \mid X \geq \Delta + 1) \neq F = \Pr(X \leq x)$, as leases
may terminate after one month (we assume $\Pr(X = 0) = 0$, though this need not
be the case in general applications). Thus, in the proceeding, all inference
about $X$ must be made from $F_0$.  This is the case in nearly all
data subject to random left-truncation, a subtle and perhaps overlooked nuance of
estimating distribution functions from left-truncated data.  For additional details,
see the seminal work \citet{woodroofe_1985}, or, for a discrete-case focused
discussion, \citet{lautier_2021}.  We also illustrate the difference between
$F_0$ and $F$ in the simulation study of Section~\ref{sec:sim_study}.

\subsection{Right-censoring}
We now introduce right-censoring. Let
$m + \Delta + 1 \leq \varepsilon \leq m + \omega$ be
the present time, at which there remain leases in the trust with ongoing
payments.  This present time, $\varepsilon$, represents the right-censoring
or pricing time. Specifically, $\Pr(T + X > \varepsilon) > 0$ and so
\begin{align*}
    X + T \leq \varepsilon & \iff X \leq m + \Delta + 1 - T + \varepsilon
    -(m + \Delta + 1)\\
    & \iff X \leq Y + \varepsilon - (m + \Delta + 1).
\end{align*}
If we define $C = Y + \varepsilon - (m + \Delta + 1)$, then it is clear
the right-censoring time is a function of the left-truncation random variable $Y$.
More precisely, $C$ equals the left-truncation time $Y$ plus a constant.
As such, it is convenient to define
\begin{equation*}
    \tau = \varepsilon - (m + \Delta + 1),
\end{equation*}
and so $C = Y + \tau = \varepsilon - T$.  If $\varepsilon > \omega + m$,
then there are no right-censored observations.

Consider now the observable range of $X$.  In the case of no left-truncation
and no right-censoring, it is clear $1 \leq X \leq \omega$; that is, the entire
distribution of $X$ is observable.  In the case of left-truncation, each lease
in the trust will have a minimum survival time of $\Delta + 1$, and so
$\Delta + 1 \leq X \leq \omega$, as we demonstrated in the previous section.
If we also include right-censoring, then each lease termination time will only be
observable if $X \leq C = Y + \tau = \varepsilon - T$. Hence, our observable
range of $X$ becomes  $\Delta + 1 \leq X \leq \min(\omega, \varepsilon - 1)$.
It is convenient to write $\xi = \min(\omega, \varepsilon - 1)$, and so
$\Delta + 1 \leq X \leq \xi$. That is, the complete finite right
tail of $X$ is estimable only if $\varepsilon - 1 \geq \omega$.  On the other
hand, if $\varepsilon - 1 < \omega$, then there is no information on the
distribution function of $X$ for $x \in \{\varepsilon, \ldots, \omega\}$.
Figure~\ref{fig:number_line} summarizes the three possible lease lifetime
data outcomes as of time $\varepsilon$: left-truncated, complete, and
right-censored.


\begin{figure}[tbp]
        \centering
		\begin{tikzpicture}
		\begin{axis}[
  		height=6cm,
  		width=15cm,
  		axis y line=none,
  		axis lines=left,
  		axis line style={-},
  		xmin=1,
  		xmax=24,
  		ymin=-5,
  		ymax=48,
  		xtick={1, 3.5, 6, 9, 12.5, 18,24},
  		xticklabels={1, $T$, $m$, $m + \Delta$, $m+\Delta+1$, $\varepsilon$, $m+\omega$},
		]
		\addplot[color = black,mark=none,thick] expression[domain=3.66:7.34,samples=100] {15};
		\addplot[black, mark = o, mark size=2.9pt, mark options={fill=white}] coordinates {(3.5, 15)};
		\addplot[black, mark = otimes, mark size=2.9pt, mark options={fill=black}] coordinates {(7.5, 15)};
		\addplot[black,mark=none,thick] expression[domain=3.66:14.95,samples=100] {25};
		\addplot[black, mark = square, mark size=2.9pt, mark options={fill=white}] coordinates {(3.5, 25)};
		\addplot[black, mark = *,mark options={fill=black}] coordinates {(15, 25)};
		\addplot[color = black,mark=none,thick] expression[domain=3.63:18,samples=100] {35};
		\addplot[black, mark = triangle, mark size=3.5pt, mark options={fill=white}] coordinates {(3.5, 35)};
		\addplot[black, mark = |, mark size=2.9pt, mark options={fill=black}] coordinates {(18, 35)};
		\draw[black, decorate, decoration={brace, amplitude=5pt, raise=1pt}]
  		(axis cs:3.5,-3) -- (axis cs:12.5,-3) node[font=\small, pos=.5, above=5pt] 
  		{$Y$};
  		\draw[black, decorate, decoration={brace, amplitude=5pt, raise=1pt}]
  		(axis cs:3.5,5) -- (axis cs:18,5) node[font=\small, pos=.5, above=5pt] 
  		{$C$};
		\end{axis}
		\end{tikzpicture}
		\caption{Three possible lease origination lifetime data 
		outcomes at current
		time $\varepsilon$ originated at time $T$.
		Left-truncated $(\iscircle)$: 
		A lease originated at time $T$ does not
		survive until $Y = m + \Delta + 1 - T$.
		Such an outcome would not be observable to an investor. 
		Complete $(\square)$: A lease originated at time $T$ survives
		longer than $Y = m + \Delta + 1 - T$ and terminates prior to time
		$C = Y + \varepsilon - (m + \Delta + 1) \equiv Y + \tau$. 
		The complete lifetime, $X$, (the length of the line segment from
		$\square$ to $\bullet$) is observable to the investor
		(though still conditional on surviving at least $Y$ months).
		Right-censored $(\triangle)$: A lease originated at time $T$ is
		still active as of time $\varepsilon$.
		The investor observes
		$X \geq C = Y + \varepsilon - (m + \Delta + 1) \equiv Y + \tau$
		but does not observe the exact termination time, $X$.}
    \label{fig:number_line}
\end{figure}
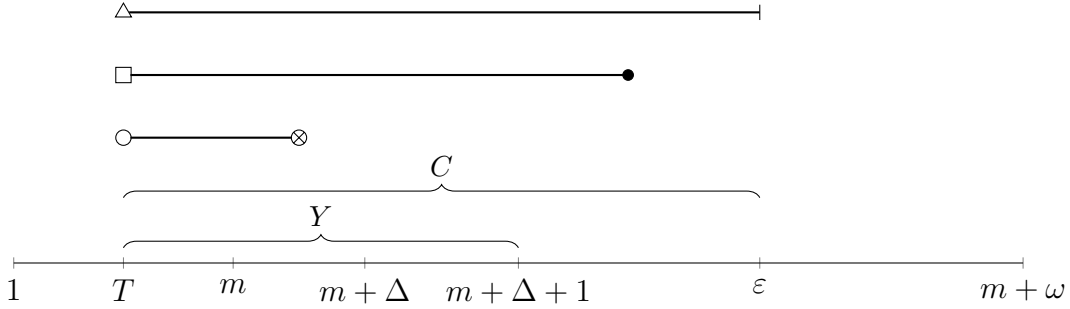

\subsection{Sampling}

We now have a description of our lifetime variable of interest, $X$, the
left-truncation random variable, $Y$, and the right-censoring random variable,
$C = Y + \tau$. In an applied setting, of course, we have observed data and,
from this data, must attempt to infer information about $X$.  Thus, how such
data may be generated or sampled from some population of independent random
vectors $(X,Y)$ such that $X$ is independent of $Y$ is of interest.
Since the securitized trust consists of
only those pairs of $(X,Y)$ such that $X \geq Y$, we assume our population
has already been left-truncated.  Hence, it is this left-truncated
population from which we are sampling $(X_i, Y_i)$ for $1 \leq i \leq n$.
Given the machinations of the securitization process,
this is more appropriate for our application than the assumed sampling process
of \citet{woodroofe_1985} or \citet{tsai_1987}, which samples $(X_i, Y_i)$ for
$1 \leq i \leq n'$, $n' \geq n$
and then removes all pairs $(X_i, Y_i)$ if $Y_i > X_i$.
Phrased differently, our process effectively samples from the
already left-truncated lease data within the trust rather than imagines we are able
to sit with the ABS issuer and see loans that did not meet the minimum lifetime
age to be included in the trust.  A theoretical divergence with limited
practical significance in most applications,
but it does indeed emerge with ABS data.

Finally, because of right-censoring, we do not observe $(X_i, Y_i)$.  Instead, for
each pair $(X_i, Y_i)$, we observe only the random variables $Y_i$,
$\min(X_i, C_i)$, where $C_i = Y_i + \tau$, and if $X_i$ was right-censored
(typically denoted by an indicator function, $\mathbf{1}_{X_i \leq C_i}$,
where $\mathbf{1}_{X_i \leq C_i} = 1$ if $X_i \leq C_i$ and equals 0
otherwise).  Our next goal, therefore, is to extract as much
information as possible about $X$ from these three random variables.

\section{Estimation}
\label{sec:est}

As is the usual approach in survival analysis, it is convenient to work in
terms of the hazard rate.  Since our application requires discrete-time on the
set of integers, we will use the following hazard rate definition
\begin{equation}
\lambda(x) = \frac{ \Pr(X = x) }{ \Pr(X \geq x) }.
\label{eq:haz_rate}
\end{equation}
The lifetime random variable, $X$, is often assumed to be continuous, and
so some readers may be more familiar with \eqref{eq:haz_rate}
expressed as a limit \citep[e.g.,][equation (2.3.1), pg. 27]{klein_2003}.
From \eqref{eq:haz_rate}, we can recover the distribution function for $X$
through
\begin{equation*}
1 - F(x-) = \Pr(X \geq x) = \prod_{\Delta + 1 \leq k < x} [1 - \lambda(k)].
\end{equation*}

It is enough, therefore, is to estimate~\eqref{eq:haz_rate}.  In building to
the discrete-time product-limit estimator subject to random left-truncation and
right-censoring, recall again our observable data: for each lease,
$1 \leq i \leq n$,
we observe the triple $(Y_i, \min(X_i, C_i), \mathbf{1}_{X_i \leq C_i})$,
where $C_i = Y_i + \tau$.
Since $\tau$ is a constant,
which depends on the pricing time, $\varepsilon$, and the fixed times $m$ and
$\Delta$, $\tau$ is independent of any specific lease $i$, $1 \leq i \leq n$
(this a convenient property of our application, see
Section~\ref{subsec:understand} for additional details).  In other words, the
observable triple derives from the pair $(X_i, Y_i)$, $1 \leq i \leq n$, which
is an independent copy of $(X,Y)$ given $X \geq Y$ and under the assumption 
that $X$ and $Y$ are independent, and the constant, $\tau$.

In the following, the subscript $\tau$
will indicate an underlying data set that has been left-truncated and
right-censored. Define
$\alpha = \Pr(X \geq Y)$,
\begin{align}
f_{*, \tau}(x) = \Pr(X_i = x, X_i \leq C_i)
	&= \Pr(X = x, X \leq C \mid X \geq Y) \nonumber\\
    &= \frac{ \Pr(X = x) \Pr(Y \leq x \leq C) }{ \alpha }, \label{eq:f_star}
\end{align}
and
\begin{align}
U_{\tau}(x) = \Pr(Y_i \leq x \leq \min(X_i, C_i))
    &= \Pr(Y \leq x \leq \min(X,C) \mid X \geq Y) \nonumber\\
    &= \frac{ \Pr(Y \leq x \leq C) \Pr(x \leq X) }{ \alpha }
    \label{eq:c_tau}.
\end{align}
Therefore,
\begin{align*}
\frac{ f_{*, \tau}(x) }{ U_{\tau}(x) }
    =\bigg[ \frac{ \Pr(X = x) \Pr(Y \leq x \leq C) }{ \alpha } \bigg]
    \bigg[ \frac{ \alpha }{ \Pr(Y \leq x \leq C) \Pr(x \leq X) } \bigg]
    = \lambda_{\tau}(x).
\end{align*}
See Section 2 of \cite{lautier_2021} for an extended discussion of why having
$U_{\tau}(x)$ in the denominator is not a concern (i.e., it is nonzero).

\begin{remark}
We have been assuming $C = Y + \tau$, where
$\tau = \varepsilon - (m + \Delta + 1)$, a constant.  However, the results
hold more generally if $C=f(Y)$, where $f$ is a Borel function and $C \geq Y$
almost surely.
\end{remark}

Since \eqref{eq:f_star} and \eqref{eq:c_tau} are directly estimable from
the data via
\begin{equation*}
\hat{f}_{*, \tau, n}(x) =
\frac{1}{n} \sum_{i=1}^{n}
\mathbf{1}_{X_i \leq C_i} \mathbf{1}_{\min(X_i,C_i) = x},
\quad \text{and} \quad
\hat{U}_{\tau,n}(x) =
\frac{1}{n} \sum_{i=1}^{n} \mathbf{1}_{Y_i \leq x \leq \min(X_i,C_i)},
\end{equation*}
we have the natural estimator for \eqref{eq:haz_rate} as follows:
\begin{equation}
\hat{\lambda}_{\tau,n}(x) =
\frac{ \hat{f}_{*, \tau, n}(x) }{ \hat{U}_{\tau,n}(x)} =
\frac{ \sum_{i=1}^{n}
\mathbf{1}_{X_i \leq C_i} \mathbf{1}_{\min(X_i,C_i) = x} }
{ \sum_{i=1}^{n} \mathbf{1}_{Y_i \leq x \leq \min(X_i,C_i)}}.
\label{eq:est_haz}
\end{equation}
The discrete-time point estimator we have derived under the
preliminary conditions of Section~\ref{sec:prelim} in \eqref{eq:est_haz}
coincides directly with \citet{tsai_1987}.  Notably, under ``minor technical
restrictions" on the support space of $X$ and $Y$, \citet{tsai_1987} state that
\eqref{eq:est_haz} is the nonparameteric conditional maximum likelihood
estimator of $\lambda_{\tau}$.
Further, it is not difficult to show that \eqref{eq:est_haz} is the same
point estimator as in Section 18.4.3 of \cite{dickson_2020} and Section 12.1 of
\citet{klugman_2012}.  We prefer the indicator representation because of its
natural relationship to computational programming, however, which facilitates
applications.

Also of interest is the asymptotic properties of \eqref{eq:est_haz}.  In
\citet{tsai_1987}, the authors provide the asymptotic properties of
\eqref{eq:est_haz}, but they assume a continuous survival function in doing so.
Hence, we cannot apply the asymptotic results of \citet{tsai_1987} to the
discrete space we carefully defined in Section~\ref{sec:prelim}.
Our main theoretical contribution is thus the asymptotic properties of the
vector of estimators
$\hat{ \bm{\Lambda} }_{\tau,n} =
(\hat{\lambda}_{\tau,n}(\Delta + 1),
\ldots,
\hat{\lambda}_{\tau,n}(\xi))^{\top}$
under the discrete assumptions of Section~\ref{sec:prelim}.
Specifically, we show that $\hat{ \bm{\Lambda} }_{\tau,n}$ is asymptotically normal
with independent components and unbiased for the vector of true hazard rates,
$\bm{\Lambda}_{\tau} =
(\lambda_{\tau}(\Delta + 1), \ldots, \lambda_{\tau}(\xi))^{\top}$.  We state this
formally in Theorem~\ref{thm:haz_norm} and provide a complete proof in
Appendix~\ref{sec:asym_proof}.

\begin{theorem}[$\hat{\bm{\Lambda}}_{\tau,n}$ Asymptotic Properties]
Define $\hat{\bm{\Lambda}}_{\tau,n} = \big( \hat{\lambda}_{\tau,n}(\Delta+1),
\ldots, \hat{\lambda}_{\tau,n}(\xi)\big)^{\top}$, where $\hat{\lambda}_{\tau,n}$
follows from \eqref{eq:est_haz}.  Then,
\begin{enumerate}[label=(\roman*)]
	\item \begin{equation*}
			\hat{\bm{\Lambda}}_{\tau, n}
			\overset{\mathcal{P}}{\longrightarrow}
			\bm{\Lambda}_{\tau}, \text{ as } n \rightarrow \infty;
		  \end{equation*}
	\item \begin{equation*}
    			\sqrt{n}( \hat{\bm{\Lambda}}_{\tau,n} - \bm{\Lambda}_{\tau})
    			\overset{\mathcal{L}}{\longrightarrow} N(0, \bm{\Sigma}),
    			\text{ as } n \rightarrow \infty,
		  \end{equation*}
\end{enumerate}
where $\bm{\Lambda}_{\tau} =
\big( \lambda_{\tau}(\Delta + 1), \ldots, \lambda_{\tau}(\xi)\big)^{\top}$ with
$\lambda_{\tau} = f_{*,\tau} / U_{\tau}$ and
\begin{equation*}
    \bm{\Sigma} = \textup{diag} \bigg(
    \frac{ f_{*, \tau}(\Delta+1) \{U_{\tau}(\Delta+1) - f_{*,\tau}(\Delta+1)\}}
    {U_{\tau}(\Delta+1)^3}, \ldots,
    \frac{ f_{*, \tau}(\xi) \{U_{\tau}(\xi) -
    f_{*,\tau}(\xi)\}}{U_{\tau}(\xi)^3} \bigg).
\end{equation*}
That is, the estimators $\hat{\lambda}_{\tau,n}(\Delta + 1), \ldots,
\hat{\lambda}_{\tau,n}(\xi)$ are consistent, asymptotically normal,
and independent.
\label{thm:haz_norm}
\end{theorem}
Under suitable conditions, Theorem~\ref{thm:haz_norm} allows risk managers
to account for the variability of the estimators $\hat{\bm{\Lambda}}_{\tau,n}$,
and we provide an illustrative example in Section~\ref{subsec:quant_est}.
Lastly, we again wish to emphasize that we can make meaningful inference about
$F$ from $F_0$ only; we cannot recover $F$.  We demonstrate the effects of
left-truncation and right-censoring on our ability to recover the tails of $X$
(as well as the asymptotic properties of Theorem~\ref{thm:haz_norm})
in the simulation study of Section~\ref{sec:sim_study}.  Meaningful
inference on $X$ is still possible, however, which is the main contribution of
the related statistical literature.

\section{Cash flow model}
\label{sec:cash_flow}

We first introduce the model within the context of a consumer auto-lease
asset-backed security.  Next, we provide the major financial result of this
work in the formulaic estimators for an expected present
value of a pool of time-to-event contingent contracts over a monthly time horizon
of the investor's choice.  The section closes with a digression to emphasize
the model's important assumptions and considerations before generalizing it to
other forms of ABS.  To assess its performance in a realistic application,
readers may proceed to Section~\ref{sec:app}.

\subsection{Pricing model}

Our objective is to calculate the present value or price of future cash flows
from a trust of consumer automobile lease contracts.  For generality, suppose
the present time is
$m + \Delta + 1 \leq \varepsilon \leq m + \omega$.  This implies the trust is
ongoing with payment history but not yet terminated.  We will elucidate our model
by building up from a single lease contract to the complete trust.  Recall
there are $n$ total lease contracts at origination and define $n_{\varepsilon}$ 
to be the number of active lease contracts at time~$\varepsilon$.
Naturally, $n_{\varepsilon} \leq n$.
We consider a lease $i$ that is still active and paying at time
$\varepsilon$, where $1 \leq i \leq n_{\varepsilon}$.

Suppose that the age of this lease contract $i$ at time $\varepsilon$ is
$\Delta + 1 \leq x_{\varepsilon(i)} \leq \xi$.  For lease $i$, denote the
monthly contractual payment as $c_i$, the contract residual value as $v_i$, the
random month of termination as $X_i$ where
$x_{\varepsilon(i)} \leq X_i \leq \xi$, the $k$th month
spot rate as $r_k$ for $1 \leq k \leq X_i - x_{\varepsilon(i)} + 1$, the sale
time multiplicative scalar random variable given $X_i$ as $Z_{X_i}$, and the
random number of months past the point of the final monthly lease payment until
the vehicle is sold and the trust is repaid given $X_i$ as $D_{X_i}$, where
$D_{X_i}$ is an integer over the range
$\{0, 1, \ldots, \min(d_{\text{max}}, X_i - x_{\varepsilon(i)} + 1)\}$.
Here, $d_{\text{max}} \leq \xi$ is a finite positive integer.  Define
the present value of the monthly lease payments as
\begin{equation*}
    W_i(X_i, D_{X_i})
    = \sum_{j=1}^{X_i - D_{X_i} - x_{\varepsilon(i)} + 1}
    \frac{c_i}{(1+r_j)^j},
\end{equation*}
and the present value of the contractual residual payment as
\begin{equation*}
    R_i(X_i) =
    \frac{ v_i }
    {(1 + r_{X_i - x_{\varepsilon(i)}+1})^{X_i - x_{\varepsilon(i)}+1}}.
\end{equation*}
Then, the present value (PV) at time $\varepsilon$ of the future payments
for lease $i$ is
\begin{equation}
    \text{PV}_i = W_i(X_i, D_{X_i}) + R_i(X_i) Z_{X_i}.
    \label{eq:PV}
\end{equation}

\begin{remark}
Let us connect \eqref{eq:PV} to the inherent lessee optionality embedded in
an automobile lease contract.  In the event the lessee elects to purchase the
vehicle at contract termination, $D_{X_i} = 0$, and so there is no gap
between the final monthly payment and the large residual payment.  In the
event $D_{X_i} = 0$, then $Z_{X_i} \approx 1$, as the purchase price is likely
very close to~$v_i$.
On the other hand, if the lessee declines to purchase the vehicle, the dealer
must sell the automobile to repay the trust.  In this case, we expect some
delay and so $D_{X_i} > 0$.  Further, it is also likely in this case
$Z_{X_i} \neq 1$.
In this sense, given a lease contract termination time of~$X_i$, we can
interpret $\Pr(D_{X_i} = 0)$ as the probability a lessee elects to purchase
the automobile at contract termination.
\end{remark}

We assume all payments are received at the end-of-the-reporting-period.
The quantities $c_i$, $v_i$, and $x_{\varepsilon(i)}$ are known for lease $i$.
Within a lease contract, the purchase price
of the automobile is set at onset.  In total, the randomness of
$\text{PV}_i$ follows from the randomness of the lease contract termination
time, $X_i$, the random residual realization, $Z_{X_i} v_i$, and the random
delay time between receipt of the final monthly payment and the residual,
$D_{X_i}$.
In Section~\ref{subsec:quant_est}, we will also incorporate a fourth component
of randomness in the form of the statistical estimation error of the
distribution of $X$ (valid in certain situations, see
Section~\ref{subsec:understand} for details).

It may be illustrative to connect the notation of \eqref{eq:PV} with
two realized lease contract cash flows from
the the Mercedes-Benz Auto Lease Trust (MBALT) 2017-A consumer automobile
lease asset-backed security \citep{mercedes_2017}, which will be introduced
more completely in Section~\ref{sec:app}.  As such, we have summarized two
realized lease cash flows in Table~\ref{tab:samp_life}
over the first eight months of the
securitization.  For example asset number~1, we
have $x_{\varepsilon(i)} = 30$, $c_i = 739$, $v_i = 36{,}383$, $X_i = 37$,
$D_{X_i} = 3$, and $Z_{X_i} = 30{,}690 / 36{,}383 = 0.844$.  Similarly,
for example asset number~2
we have $x_{\varepsilon(i)} = 31$,
$c_i = 678$, $v_i = 32{,}376$, $X_i = 36$, $D_{X_i} = 1$, and
$Z_{X_i} = 24{,}576 / 32{,}376 = 0.759$.

The present value of the complete trust at time $\varepsilon$
then follows from \eqref{eq:PV} as
\begin{equation}
    \text{PV}_{\text{Trust}} = \sum_{i = 1}^{n_{\varepsilon}} \text{PV}_i.
    \label{eq:pv_trust}
\end{equation}

In the following, we present steps to calculate the price of such a trust
vis-\`{a}-vis taking an expectation of \eqref{eq:pv_trust}, perhaps more
commonly known to some readers as computing the actuarial present value (APV).

\begin{center}
\begin{table}[tbh]
    \centering
    \begin{tabular}{ccccccccc}
    & \multicolumn{4}{c}{Ex. Asset Num: 1} & 
    \multicolumn{4}{c}{Ex. Asset Num: 2}\\
    \cmidrule(lr){1-1} \cmidrule(lr){2-5} \cmidrule(lr){6-9}
    Obs. Month & Age & Pmt & Resid. & Con. Resid. &
    Age & Pmt & Resid. & Con. Resid.\\
    \cmidrule(lr){1-1} \cmidrule(lr){2-5} \cmidrule(lr){6-9}
    1 & 30 & 1{,}478 & 0 & -- & 31 & 1{,}357 & 0 & --\\
    2 & 31 & 739 & 0 & -- & 32 & 678 & 0 & --\\
    3 & 32 & 739 & 0 & -- & 33 & 678 & 0 & --\\
    4 & 33 & 739 & 0 & -- & 34 & 678 & 0 & --\\
    5 & 34 & 739 & 0 & -- & 35 & 678 & 0 & --\\
    6 & 35 & 739 & 0 & -- & 36 & 0 & 24{,}576 & 32{,}376\\
    7 & 36 & 0 & 0 & -- & -- & -- & -- & --\\
    8 & 37 & 0 & 30{,}690 & 36{,}383 & -- & -- & -- & --\\
    \cmidrule(lr){1-1} \cmidrule(lr){2-5} \cmidrule(lr){6-9}
    \end{tabular}
    \caption{MBALT 2017-A sample life cash flows}
    \label{tab:samp_life}
\end{table}
\end{center}

\subsection{Expected or actuarial present value}

Throughout this section assume $r_k$ is the deterministic spot rate for
month $1 \leq k \leq \xi$.  It may be generated stochastically from an
interest rate model or other economic scenario generator, but the monthly
rate shall be treated as an user input discount assumption.

In taking an expectation of \eqref{eq:pv_trust}, we will need the probabilities
$\Pr(X_i = s \mid X \geq x_{\varepsilon(i)})$, which we denote
$p_{x_{\varepsilon(i)}}^s$, for
$x_{\varepsilon(i)} \leq s \leq \xi$.  The probabilities
$p_{x_{\varepsilon(i)}}^s$ may be defined in terms of the
hazard rate, i.e., \eqref{eq:haz_rate}.  Precisely,
\begin{align*}
    \Pr(X_i = s \mid X_i \geq x_{\varepsilon(i)})
    &= \Pr(X = s \mid X \geq x_{\varepsilon(i)})\\
    &= \lambda_{\tau}(s) \prod_{x_{\varepsilon(i)} \leq k \leq s-1}
    [1 - \lambda_{\tau}(k)],
\end{align*}
where we use the convention
\begin{equation*}
    \prod_{x_{\varepsilon} \leq k \leq s-1}
    [1 - \lambda_{\tau}(k)] = 1,
\end{equation*}
if $s = x_{\varepsilon}$.  It is possible that the number of months beyond
age $x_{\varepsilon(i)}$ that a lease contract terminates will be less than
$d_{\text{max}}$.  In this case, we will load all possible delay time
probabilities beyond $X_i - x_{\varepsilon(i)} + 1$ for a given $X_i$ onto
$\Pr(D = X_i - x_{\varepsilon(i)} + 1 \mid X = X_i)$.  Formally, we
define
\begin{equation*}
    \Pr^*(D_{X_i} = k) =
    \Pr(D = k \mid X = X_i) +
    \mathbf{1}_{k = X_i - x_{\varepsilon(i)} + 1}
    \bigg( \sum_{k = X_i + x_{\varepsilon(i)} + 2}^{d_{\text{max}}}
    \Pr(D = k \mid X = X_i) \bigg).
\end{equation*}
Finally, let us use the notation
$\varphi = \min( d_{\textup{max}}, X_i - x_{\varepsilon(i)} + 1)$. We are
thus ready to state the major result of this section, with its proof in
Appendix~\ref{subsec:thm_apv}.

\begin{theorem}
Assume the framework of Sections~\ref{sec:prelim} and \ref{sec:est}.
Suppose the present time is
$\varepsilon$, where $m + \Delta + 1 \leq \varepsilon \leq m + \omega$,
and we have a collection of $n_{\varepsilon}$ time-to-event contingent cash
flows streams that are still active and paying following the individual model
\eqref{eq:PV} and the aggregate model \eqref{eq:pv_trust}.
Call the collection of these random cash flow streams the Trust.  Denote the
lifetime random variable of interest for lease $i$,
$1 \leq i \leq n_{\varepsilon}$, by $X_i$. Then the actuarial present value \textup{(APV)}
of the Trust is
	\begin{equation}
    		\textup{APV}_{\textup{Trust}} = \sum_{i = 1}^{n_{\varepsilon}} \textup{APV}_i,
    		\label{eq:APV_trust}
	\end{equation}
	where
	\begin{equation*}
    \textup{APV}_i = \sum_{m=x_{\varepsilon(i)}}^{\xi}
    \bigg(
    \bigg{ \{ }
    \sum_{k = 0}^{\varphi }
    W_i(m,k) \Pr^*(D_{m} = k)
    \bigg{ \} }
    + R_i(m) \mathbf{E}(Z \mid X = m)
    \bigg) p_{x_{\varepsilon(i)}}^{m}.
    \end{equation*}
\label{thm:apv_varapv}
\end{theorem}

In a practical setting, the underlying distributions of the random quantities
in Theorem~\ref{thm:apv_varapv} will need to be estimated. If we assume
independence between $X$ and the left-truncation random variable $Y$ ---
a non-trivial assumption that we
discuss more fully in Section~\ref{subsec:understand} --- then the results of
Section~\ref{sec:est} may be used to estimate the recoverable portion of the
distribution for the time-until-contract-termination, $X$.  Furthermore,
as we demonstrated in Theorem~\ref{thm:haz_norm}, the estimator
\eqref{eq:est_haz} will be asymptotically unbiased.  This suggests that the
use of the estimators $\hat{\lambda}_{\tau,n}$ in place of the true hazard
rates in Theorem~\ref{thm:apv_varapv} along with asymptotically unbiased
estimates for $\displaystyle \Pr^*(D_{X_i} = k)$, $0 \leq \varphi$ and
$\mathbf{E}(Z \mid X = m)$, $\Delta + 1 \leq m \leq \xi$ (such as standard
empirical estimates), will yield a
close approximation for the true expected present value \eqref{eq:APV_trust}
of a trust of time-to-event contingent cash flows for large $n$ (for the sake
of clarity, we repeat here that $n$ is the number of leases active at the
origination of the trust at time $m + \Delta$, whereas $n_{\varepsilon}$ is
the number of leases active at time $\varepsilon$). If the
assumption of mutual independence between $(X_i, Y_i)$ and $(X_j, Y_j)$ for
$1 \leq i \neq j \leq n$ is also satisfied --- another non-trivial
assumption discussed more in Section~\ref{subsec:understand} --- then we
may also assess the inherent uncertainty of the price point estimator of
Theorem~\ref{thm:apv_varapv} (see Section~\ref{subsec:quant_est}).

In many financial applications it may be desirable to calculate
a present value for a fixed amount of time, such as over the next six months
or one year.  The equations of Theorem~\ref{thm:apv_varapv} may be easily
modified to do so.  Let the present value time horizon in months be denoted
by $\kappa$.  That is, if we desire to calculate the present value over the
next 12 months only, we would set $\kappa = 12$.  To illustrate, define
\begin{equation*}
    W_i^*(X_i, D_{X_i}) =
    \sum_{j=1}^{\min(\kappa, X_i - D_{X_i}-x_{\varepsilon(i)} + 1)}
    \frac{ c_i }{(1 + r_j)^j},
\end{equation*}
let $\mathbf{1}_{\kappa}$ be an
indicator function that equals 1 if
$\kappa \leq X_i - x_{\varepsilon(i)} + 1$ and zero otherwise, and
observe the following corollary stated without proof.

\begin{corollary}
Assume the conditions of Theorem~\ref{thm:apv_varapv}. Then the
\textup{(APV)} of the Trust over the next $\kappa$
months only, where $\kappa \in \mathbb{N}$, is
\begin{equation}
    		\textup{APV}_{\textup{Trust}}^{\kappa} =
    		\sum_{i = 1}^{n_{\varepsilon}} \textup{APV}_i^{\kappa},
    		\label{eq:APV_trust_kappa}
	\end{equation}
	where
	\begin{equation*}
    \textup{APV}_i^{\kappa} = \sum_{m=x_{\varepsilon(i)}}^{\xi}
    \bigg(
    \bigg{ \{ }
    \sum_{k = 0}^{\varphi }
    W_i^*(m,k) \Pr^*(D_{m} = k)
    \bigg{ \} }\\
    + \mathbf{1}_{\kappa}
    R_i(m) \mathbf{E}(Z \mid X = m)
    \bigg) p_{x_{\varepsilon(i)}}^{m}.
    \end{equation*}
\label{cor:cap_APV}
\end{corollary}

\subsection{Remarks on assumptions and generalizations}
\label{subsec:understand}

We digress to discuss the plausibility of two important assumptions of
independence underlying the results of Sections~\ref{sec:prelim},
\ref{sec:est}, and~\ref{sec:cash_flow}, and the potential to generalize the
model to other securitization asset classes outside of consumer
automobile-lease ABS.

We begin with the two assumptions of independence.
The first important assumption was introduced in
the estimation framework of Section~\ref{sec:est} and corresponds to
the independence between the lifetime random variable, $X$, and the
left-truncation random variable, $Y$.  If we recall our motivation, however,
$Y$ is a shifted random variable stemming from $T$, which is the origination
time random variable of an auto lease contract (see
Figure~\ref{fig:number_line} as needed).  So, the first question is
this: is it reasonable to assume $X$ and $T$ are independent?  We believe the
answer is affirmative for two reasons.  First, the total sample space of $T$
is generally over a relatively short time period, (e.g., less than three years
--- see Section~\ref{sec:app}).  Thus,
ceteris paribus, it is unlikely that a lease originated a short time away
from a second lease would have a materially different lifetime distribution.
Second, issuers of asset-backed securities are using securitization as a
financing tool for business needs (i.e., writing loans and leases to sell
more cars).  Hence, the decision of when to issue an ABS is typically driven
by market factors that are connected to the parent company, such as financing
needs and current market rates, rather than connected to the underlying
performance of the leases.  Indeed, most standard techniques of estimating
a survival curve require independence between the lifetimes of interest, the
right-censoring random variable, and the left-truncation random variable
\citep[][Chapter 3]{klein_2003}.  In our specific application, however, we
may use the independence between $T$ and $X$ to achieve independence between
$X$ and $Y$, and from the independence of $X$ and $Y$ follows the independence
between $X$ and $C$ (see Section~\ref{sec:prelim} as needed).  This is a
potentially unique advantage to our financial application that may not be
common in widespread applications, and we emphasize again the statistical
integrity of the estimation framework collapses if we lose independence
between $X$ and $Y$.

The second important assumption is independence between two given leases
within the trust to estimate the lease lifetime distribution with
\eqref{eq:est_haz} for the cash flow model of Section~\ref{sec:cash_flow}.
More precisely, this is independence between the pairs $(X_i, Y_i)$ and
$(X_j, Y_j)$ for $1 \leq i \neq j \leq n$. In some sense, this assumption is
more difficult to parse. To explain, we paraphrase the famous opening line
of Tolstoy's masterpiece \textit{Anna Karenina}: all independent
random variables are independent in the same way, but all dependent random
variables are dependent in their own way.  In other words, how does one form
a sense of dependence between two lessees?  This question is quite difficult
to answer.  Instead, we propose to proceed by (1) identifying when assuming
independence between two lease's lifetimes may be reasonable and (2) how the
model will falter if we assume independence between two leases but there is in
actuality dependence.

In light of the subprime mortgage crisis of the late 2000s, we would advise
caution in applying our estimation procedure to a pool of subprime credit
quality borrowers (i.e., a credit score below 620 \citep{cfpb_2019}). While
consensus around the widespread failures within the subprime mortgage crisis
is that mortgage defaults were largely driven by a deterioration in borrower
credit quality that was masked by appreciating home values rather than from
faulty assumptions within consumer credit models \citep{demyanyk_2009}, it
is reasonable to presume that a large scale economic event would lead to an
increased level of default in subprime borrowers.  But what of high-credit
quality borrowers leasing high-end luxury cars, such as those in the
MBALT 2017-A lease pool that is our focus in Section~\ref{sec:app}?
In this case, we feel it is reasonable to assume these lessees generally
operate with mutual economic independence. To
justify this claim, note that net losses as a percentage of average dollar
amount of lease contracts outstanding for the Mercedes-Benz aggregate lease
portfolio did not exceed 0.40\% between 2016 and the first three months of
2021, which includes the economic event of the Coronavirus pandemic
\citep{mercedes_2021}.

What if there is dependence between the pairs $(X_i, Y_i)$ and
$(X_j, Y_j)$ for $1 \leq i \neq j \leq n$, but we have incorrectly assumed
independence?  As long as the assumption of independence between the
lifetime of interest, $X$, and the left-truncation random variable, $Y$
remains valid for all $i \in \{1, \ldots, n\}$ and the dependence among
the observations is weak enough to allow the Central Limit Theorem to work,
we expect that the asymptotic unbiasedness of the vector of point estimators,
$\hat{ \bm{\Lambda} }_{\tau, n}$, to remain, and,
subsequently, so does the asymptotic unbiasedness of the point estimator
in Theorem~\ref{thm:apv_varapv} (or
Corollary~\ref{cor:cap_APV}). What changes is the uncertainty level or variance
of these estimators (e.g., $\bm{\Sigma}$ in Theorem~\ref{thm:haz_norm}).
In practice, when the sample size is big enough, the uncertainty in the
estimation of distribution will be negligible compared to the uncertainty of the
distribution itself, so its effect on a pricing exercise will be minimal.

We close this section with discussion of the appropriateness of applying
our model to other forms of ABS.  That is,
our focus is consumer automobile-lease asset-backed securities;  can the
model generalize to other asset classes of ABS?
We feel for certain asset classes, the answer is affirmative.  Certainly,
for other types of leased assets, such as equipment lease ABS, the model
extends naturally.  Other generalizations are possible with minimal
changes to the model of Section~\ref{sec:cash_flow}.  For example,
a substantial portion of asset-backed securities fall under the category
of agency mortgage-backed securities (MBS), which are issued by the
government-sponsored entities
Fannie Mae, Freddie Mac, and Ginnie Mae.  The total issuance of such debt
is over \$5 trillion, which is approximately 10\% of all credit market debt
in the United States \citep{tuckman_2012}, and of which life insurers hold
nearly \$250 billion agency MBS \citep{mcmenamin_2013}. Investors do not
bear risk of non-payment of principal or interest on the full outstanding
balance. Instead, the major risk to investors is related to the timing of
cash flows \citep{davidson_2014}.  In other words, the major risk of agency
asset-backed securities connects directly to the lifetime random variable,
$X$, as defined within \eqref{eq:PV}, and so we feel an application of our
model to agency MBS is appropriate.

When modeling a consumer loans more generally, such as student loans,
auto loans, residential mortgages, or credit cards, however,
it may be preferable to treat prepayments and defaults separately. As
currently constructed, the model of Section~\ref{sec:cash_flow} can only
handle one time-to-event outcome.  That said, we believe it may be generalized
to a competing risk environment by updating the estimator in
\eqref{eq:est_haz} and editing the details of \eqref{eq:PV} for the
various competing risk outcomes (i.e., receipt of outstanding balance
in a prepayment or recovery given default), though the details of which
remain an ongoing investigation.

\section{Simulation Study}
\label{sec:sim_study}

We present a simulation study for two purposes.  First, we will
demonstrate that the full distribution for a lifetime random variable may
not be recoverable if the estimation is performed using incomplete data.
Formally, we will see the distinction between $F$ and $F_0$, as
previously discussed in Section \ref{sec:prelim}, which can have
meaningful implications for cash flow analysis.  In some
instances, however, the contractual terms of the underlying financial
products may be used to make assumptions to partially address the
challenges of incomplete data.  We provide an example of this within our
simulation study as an illustration.  The second purpose of
out simulation study is to verify the asymptotic properties stated in Theorem
\ref{thm:haz_norm} (the complete proof may be found in
Appendix~\ref{sec:asym_proof}).

Let the lease origination random variable $T$ follow
a discrete uniform random variable over the integers
$\mathcal{T} = \{1, \ldots, 10\}$.  Using the notation from
Section~\ref{sec:prelim}, therefore, $m = 10$.  Let $\Delta = 2$ and so the the
left-truncation random variable $Y$ is discrete uniform over the integers
$\mathcal{Y} = \{3, \ldots, 12\}$.  We proceed as the though the
current time is $\varepsilon = 20$, which implies
$\tau = \varepsilon - (m + \Delta + 1) = 7$ (see Figure~\ref{fig:number_line}
as needed).
Consider now the lifetime of interest random variable, $X$, which
follows a left-truncated geometric distribution over
$\mathcal{X} = \{1, \ldots, 24\}$ with probability mass
function (pmf)
\begin{equation}
    \Pr(X = x) = \begin{cases}
    p(1-p)^{x-1}, & x = 1, 2, \ldots, 23;\\
    \sum_{x=24}^{\infty}p(1-p)^{x-1}, & x = 24;\\
    0, & \text{otherwise},
    \end{cases}
    \label{eq:trunc_geom_X}
\end{equation}
where $0 < p < 1$.  For this pmf, we set $p = 0.2$ and so $\lambda(x) = 0.2$
for $1 \leq x \leq 23$ and $\lambda(x) = 1$ for $X = 24$.
(To aid readers attempting to reproduce these results,
we calculated $\alpha = 0.2856$.)
The pmf in \eqref{eq:trunc_geom_X} implies $\omega = 24$, and so the complete
distribution of $X$ is not recoverable.  In other words, because of left-truncation
and right-censoring, we may only form estimates for
$\Delta + 1 = 3 \leq X \leq 19 = \xi = \min(\omega, \varepsilon - 1)$.

In an application to financial data, however, we may be able to reference
contractual terms that provide the basis to infer a minimum value of $\omega$,
from which we can make a reasonable assumption.  For example, in estimating the
lifetime random variable for a 24-month lease contract, we should assume that
$\omega \geq 24$.  Thus, to continue with this example, since $\xi = 19 < 24$,
we suggest to extend the estimated hazard rate for $\xi$ forward
geometrically until $24$, at which point the hazard rate should then be assumed
to be unity.  Extending the last observation assuming a geometric tail is a
common practice in survival analysis
\citep[see, for example, a discussion in the continuous case in Section 12.1 of][]{klugman_2012}.
We generated $n = 10{,}000$ pairs of
left-truncated random variables using \eqref{eq:H}.  We then calculated
$\hat{\lambda}_{\tau,n}$ using \eqref{eq:est_haz} for
$X \in \{3, \ldots, 19\}$.  This process was repeated for $1{,}000$
replicates.

Figure~\ref{fig:sim_result} shows the average of the estimated hazard function
from the $1{,}000$ replicates.  We can see
that $\hat{\bm{\Lambda}}_{\tau,n}$ (dashed line) is very close to the true
$\bm{\Lambda}_{\tau} = (0.2, \ldots, 0.2)^{\top}$ (solid line) for the
recoverable range of $X$, $3 \leq X \leq 19$.
We also plot the average 95\% true log-scale confidence interval using
Theorem~\ref{thm:haz_norm} by the shaded ribbon, its estimate using the
estimators $\hat{f}_{*, \tau, n}$ and $\hat{U}_{\tau,n}$ in place of
$f_{*,\tau}$ and $U_{\tau}$, respectively, in $\bm{\Sigma}$ by the dashed line,
and the
empirical 2.5th and 97.5th percentiles of the 1{,}000 replicates of
estimators for each recoverable $X$.  All three closely agree.  We also see
the hazard rate (and thus the associated probabilities) for
$X = 1, 2, 20, 21, 22, 23, 24$ are not recoverable due to left-truncation and
right-censoring.

\begin{figure}[tbh]
    \centering
    \includegraphics[width=\textwidth]{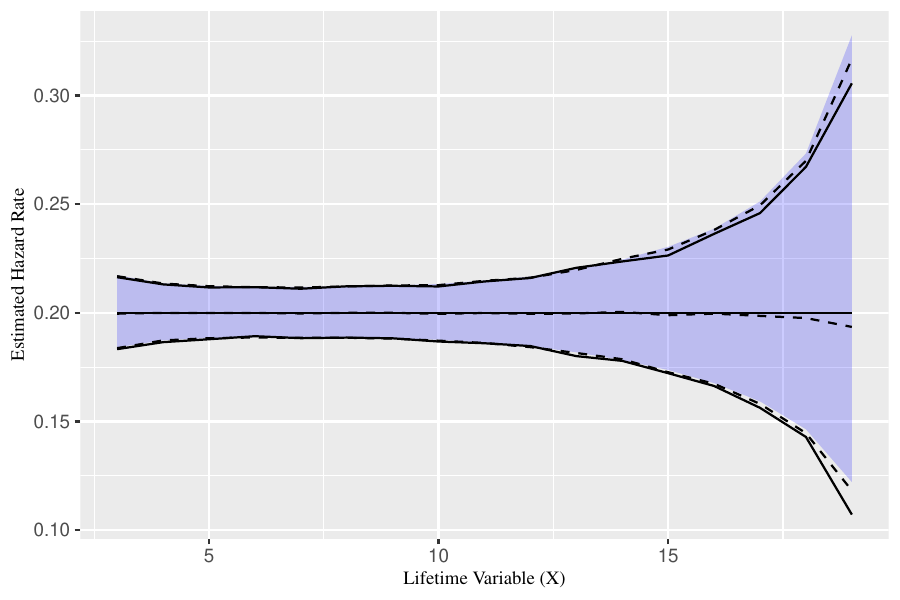}
    \caption{Theorem~\ref{thm:haz_norm} left-truncation and right-censoring
    simulation study results: true quantities (solid line), empirical
    averages of 1{,}000 replicates of hazard rate estimates and corresponding
    empirical 95\% log-scale confidence intervals (dashed lines), and the
    true 95\% log-scale confidence interval using Theorem~\ref{thm:haz_norm}
    (blue ribbon).  All three closely agree.}
    \label{fig:sim_result}
\end{figure}

We also verified Theorem~\ref{thm:apv_varapv} and Corollary~\ref{cor:cap_APV}
through a simulation study, but the results have been omitted in light of the
proof in Appendix \ref{subsec:thm_apv}.  For interested readers, please contact
the corresponding author for details.

\section{Application}
\label{sec:app}

We now demonstrate the effectiveness of the cash flow modeling and pricing
apparatus introduced in Section~\ref{sec:cash_flow} in a realistic setting.
Specifically, we will consider a large subset of lease data from
the Mercedes-Benz Auto Lease Trust (MBALT)
2017-A consumer automobile lease asset-backed security
\citep{mercedes_2017}.  Because of the aforementioned SEC Regulation AB
II enacted in November 2016 \citep{reg_ab2}, investors may now obtain
detailed loan-level borrower and monthly loan performance data through the
Electronic Data Gather, Analysis, and Retrieval (EDGAR) system maintained
by the SEC.  For additional reference, \citet{cfr_229} provides a
detailed listing of available fields.
MBALT 2017-A was placed in April of 2017 and closed in August of 2019. We
thus have 28 months of loan performance and cash flow data.  For
convenience, we have made the downloaded and cleaned data file available in
the online supplemental material.

The MBALT 2017-A transaction originally contained 56{,}402 lease
contracts on Mercedes-Benz automobiles with original terms ranging
from 24 to 60 months, the vast majority of which are 36 month leases
(47{,}315). The credit profile of a substantial portion of underlying
lessees is \textit{super-prime}, which refers to a consumer credit score
above 720 \citep{cfpb_2019}.  Nearly the entire pool of lessees would be
classified as a \textit{prime} credit, which refers to a consumer credit
score above 660 \citep{cfpb_2019}.  Further, the majority of automobiles
represent high-end or luxury vehicles.  To see this, consider that the
original vehicle value averages over just over \$61{,}600.
For a density plot of the lessee credit profile and vehicle values,
see Figure~\ref{fig:m17_cs_vv}.  As another indicator
of the low credit risk of this transaction,
between 2012 and 2016, net losses as a percentage of average dollar amount
of lease contracts outstanding for the Mercedes-Benz lease portfolio have
ranged between 0.19\% and 0.27\% \citep{mercedes_2017}.  We thus feel the
MBALT 2017-A asset-backed security is a good candidate for the proposed
model in Section~\ref{sec:cash_flow}, particularly in light of the
discussion in Section~\ref{subsec:understand}.

\begin{figure}[tbh]
    \centering
    \includegraphics[width=\textwidth]{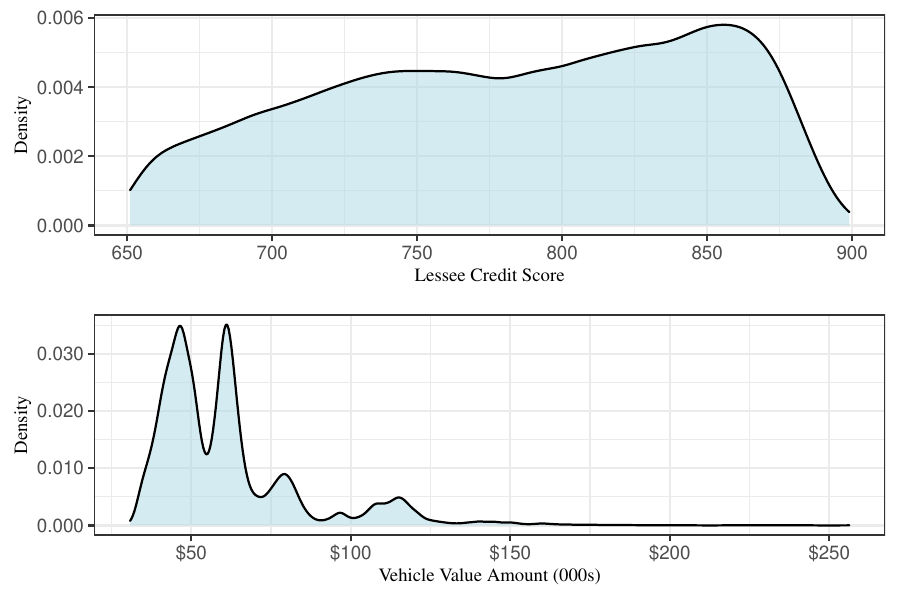}
    \caption{MBALT 2017-A credit and vehicle value distributions}
    \label{fig:m17_cs_vv}
\end{figure}

\subsection{Pricing results}
\label{subsec:price}

For the remainder of the section, we will focus on a subset of the 47{,}315
36-month leases.  Specifically, we removed 36-month lease contracts with
data irregularities that could not be easily explained from a print out of
cash flows or thorough review of data field descriptions.  The first such
irregularity was unclear or multiple residual payments.  It is our opinion that
some of the cash flows reported in the residual field
(\texttt{liquidationProceedsAmount}) represent monthly payment cash flows,
and they may be mislabeled as a form of bookkeeping convenience or error.
To evaluate the potential of the cash flow model in
Section~\ref{sec:cash_flow} and avoid introducing potentially erroneous
interpretations of the data, we have elected to remove such records. In total,
this removed 16{,}741 contracts.  In addition, there is a data field called
\texttt{terminationIndicator}, and, in an additional 729 contracts, this field
does not correspond to the month the final residual payment was received.
These records were also removed.
This leaves a total of 29{,}845 lease contracts. For an example of records
removed in the form of Table~\ref{tab:samp_life}, please see
Appendix~\ref{sec:app_detail}.
We recommend fitting a different set of hazard rate estimators
\eqref{eq:est_haz} for each original lease termination length (i.e., 24-month
leases should be fit separately from 36-month leases, and so on) because it is
prudent to assume the prior information of the termination schedule will have
a material impact on the underlying lifetime distribution.  To
price the complete trust, one may then add all different lease term groups
together.  For illustrative purposes, however, we will focus on the subset of
29{,}845 36-month lease contracts, which we will refer to as ``the Trust"
going forward.

The oldest lease in the Trust at initialization was 33 months old,
and the youngest lease was 3 months old.  Therefore, in the notation of
Section~\ref{sec:prelim}, we have $\Delta = 3$ and $m = 30$. We will denote
the lifetime random variable, $X$, as the time-until-contract-termination
(i.e., the time the final residual payment is made to the Trust, see
Table~\ref{tab:samp_life} as needed).  We will first use the methods of
Section~\ref{sec:est} to estimate the distribution for $X$. To mimic the
realities of pricing an active security, we will assume an
\textit{observation window} of 6, 12, 18, and 24 months.  By an observation
window, we mean that we will only use data from the first $\mathcal{O}$ months
of securitization payments to estimate $X$, where
$\mathcal{O} \in \{6, 12, 18, 24\}$. If we make the connection to
Figure~\ref{fig:number_line}, than an observation window of 6 months
corresponds to $\varepsilon = 39$.  Estimates for
the hazard rates of $X$ by observation window length may be found in
Figure~\ref{fig:m17_hazrate}.  We can see that the hazard rate
accelerates close to month 36, which is expected for a lease contract with a
scheduled termination of 36 months.  It is also interesting to see the effect
of left-truncation and right-censoring.  For example, we cannot recover the distribution
of $X$ before lease age 4 months.  Given the pattern of the hazard rate in
Figure~\ref{fig:m17_hazrate}, however, it is reasonable to assume the absence of
months 1-3 will have a small impact on the overall results.  In addition, we
can see that as the observation window expands and less observations are
right-censored, the right tail of $X$ gradually extends well beyond month 36.  
For the shortest observation window of 6 months, we observe terminations up
to month 38, and, since $38 \geq 36$, we assume $\xi = 38$ in this case.
Also of interest may be the resulting confidence intervals from
Theorem~\ref{thm:haz_norm}, which are denoted by the blue ribbons in
Figure~\ref{fig:m17_hazrate}.  As we can see, the extended right tails have
fewer observations and thus more estimator uncertainty than the early portions
of the distribution of $X$.

\begin{figure}[tbh]
    \centering
    \includegraphics{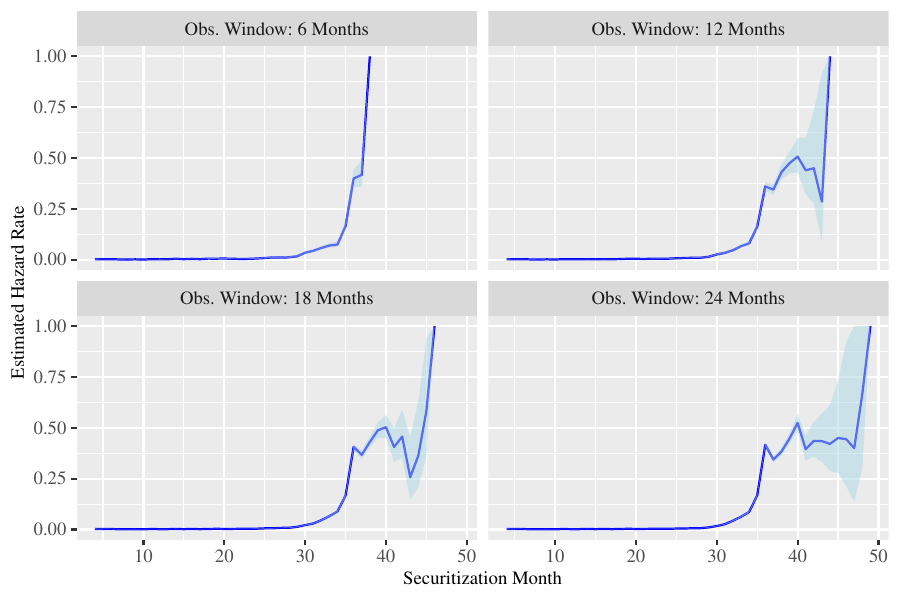}
    \caption{MBALT 2017-A hazard rate est.\ and 95\% confidence intervals
    (36-Mo.\ leases).}
    \label{fig:m17_hazrate}
\end{figure}

To determine the pricing results, we calculate empirical estimates for
$\mathbf{E}(Z \mid X = m)$ and $\displaystyle \Pr^*(D_m = k)$ for
$k \in \{0, \ldots, \varphi\}$ and $m$ spanning the recoverable sample space
of $X$, using only the observations available within the observation window.
Note that to account for the small portion of defaults, we consider the
difference between total residual realizations
(\texttt{liquidationProceedsAmount}) and default losses or
\textit{charge-offs} (\texttt{chargedOffAmount}) when estimating the final
residual payment made to the trust.
Our objective is to use Corollary~\ref{cor:cap_APV} for various values of
$\kappa$ and compare the calculated price against the actual realized
observations from the Trust.  In addition, we will also compare the results of
Corollary~\ref{cor:cap_APV} against a variation of the modeling approach used
in the MBALT 2017-A pricing prospectus \citep[][Appendix B]{mercedes_2017};
that is

\begin{quote}
Modeling Assumption: The cash flow schedules appearing in the immediately
following tables were generated assuming that (i) the lessees make their
remaining lease payments starting in March 2017 and every month thereafter
until all scheduled lease payments are made and (ii) the residual value of the
Leased Vehicles is due the month following the last related lease payment.
\citep[][pg. B-1]{mercedes_2017}
\end{quote}

Specifically, we will compare the results of Corollary~\ref{cor:cap_APV}
against an approach that uses a non-random $X$ (i.e., all leases terminate at
month 36 or immediately thereafter if older) but still uses the empirical
estimates for $\mathbf{E}(Z \mid X = m)$.  In other words, can our formulaic
approach that accounts for left-truncation and right-censoring improve upon the
non-random $X$ pricing assumptions of a newly issued auto-lease ABS?  In
Table~\ref{tab:apv_results}, we refer to non-random $X$ as the
\textit{Contract} approach.
To produce the comparisons, we utilized three interest-rate discount scenarios.
The first is a simple zero-interest scenario, the second is a standard
scenario, and the last is an inverted scenario.  The exact spot rates used in
each scenario may be found in Appendix~\ref{sec:app_detail}.

Based on the results in Table~\ref{tab:apv_results}, we can see that
Corollary~\ref{cor:cap_APV} generally outperforms the contract approach and
provides an accurate estimate of the value of future cash flows.  More
specifically, Corollary~\ref{cor:cap_APV} is quite accurate over the
short-term, even for an observation window of only 6 months.  As the
observation window increases, the results of Corollary~\ref{cor:cap_APV}
improve over a longer horizon as well, until they eventually are inside the
Contract approach for all time horizons by an observation window of 12 months.
As the observation window size increases, Corollary~\ref{cor:cap_APV} begins
to significantly outperform the Contract approach.  The results hold
generally across the three interest rate scenarios.  Though not reported within
this manuscript, we also note that the results are consistent with
Table~\ref{tab:apv_results} for a subset of 24-month leases from the MBALT
2017-A transaction.

Because the results of Corollary~\ref{cor:cap_APV} can be computed via a
formula and do not require extensive simulation, we feel the ability to improve
upon the non-random transaction prospectus approach
\citep[][Appendix B]{mercedes_2017} --- especially as the observation window
increases --- at limited additional effort to be an
advantage of our approach.  Furthermore, the ability to estimate a fully
specified asymptotic distribution of the hazard rate estimators using
Theorem~\ref{thm:haz_norm} --- under the mutual independence assumption of
$(X_i, Y_i)$ and $(X_j, Y_j)$ for $1 \leq i \neq j \leq n$, which we feel is
satisfied for the MBALT 2017-A pool of leases ---
allows for an assessment of the potential
uncertainty of the pricing point estimates of Corollary~\ref{cor:cap_APV}. As
such, we provide a demonstration in the following section.

\begin{sidewaystable}
\centering
	{ \footnotesize
    \begin{tabular}{cccccrrcccrrcccrr}
      \toprule
    \multicolumn{2}{c}{} &
    \multicolumn{5}{c}{Zero Interest Rates} &
    \multicolumn{5}{c}{Standard Yield Curve} &
    \multicolumn{5}{c}{Inverted Yield Curve}\\
      \cmidrule(lr){3-7} \cmidrule(lr){8-12} \cmidrule(lr){13-17}
 Obs. Win. &$\kappa$&Con.&APV&Act.&Con.[\%]&APV[\%]&Con.&APV&Act.&Con.[\%]&APV[\%]&Con.&APV&Act.&Con.[\%]&APV[\%]\\
      \midrule
    \multirow{8}{*}{6 mo.}
&3&75&101&102&$-26.47$&$-$0.98&74&99&101&$-$26.73&$-$1.98&70&95&97&$-$27.84&$-$2.06\\
&6&169&210&205&$-$17.56&2.44&163&203&198&$-$17.68&2.53&153&190&186&$-$17.74&2.15\\
&9&258&317&310&$-$16.77&2.26&242&298&292&$-$17.12&2.05&225&277&271&$-$16.97&2.21\\
&12&359&436&420&$-$14.52&3.81&325&395&382&$-$14.92&3.40&302&367&355&$-$14.93&3.38\\
&15&466&567&553&$-$15.73&2.53&405&493&481&$-$15.80&2.49&378&461&450&$-$16.00&2.44\\
&18&629&716&671&$-$6.26&6.71&518&596&563&$-$7.99&5.86&489&562&530&$-$7.74&6.04\\
&21&770&860&792&$-$2.78&8.59&606&687&639&$-$5.16&7.51&581&656&609&$-$4.60&7.72\\
&22&804&911&837&$-$3.94&8.84&626&717&666&$-$6.01&7.66&602&688&638&$-$5.64&7.84\\
    \midrule
    \multirow{6}{*}{12 mo.}
&3&71&95&105&$-$32.38&$-$9.52&70&94&104&$-$32.69&$-$9.62&67&90&99&$-$32.32&$-$9.09\\
&6&159&204&215&$-$26.05&$-$5.12&154&197&208&$-$25.96&$-$5.29&144&185&195&$-$26.15&$-$5.13\\
&9&258&327&348&$-$25.86&$-$6.03&242&306&326&$-$25.77&$-$6.13&224&284&302&$-$25.83&$-$5.96\\
&12&408&463&466&$-$12.45&$-$0.64&364&417&423&$-$13.95&$-$1.42&338&387&392&$-$13.78&$-$1.28\\
&15&539&598&587&$-$8.18&1.87&462&519&513&$-$9.94&1.17&431&484&479&$-$10.02&1.04\\
&16&571&645&632&$-$9.65&2.06&485&552&545&$-$11.01&1.28&453&516&510&$-$11.18&1.18\\
      \midrule
    \multirow{4}{*}{18 mo.}
&3&66&108&133&$-$50.38&$-$18.80&65&106&131&$-$50.38&$-$19.08&62&102&125&$-$50.40&$-$18.40\\
&6&194&238&251&$-$22.71&$-$5.18&187&229&243&$-$23.05&$-$5.76&174&215&228&$-$23.68&$-$5.70\\
&9&317&370&371&$-$14.56&$-$0.27&296&347&351&$-$15.67&$-$1.14&273&322&326&$-$16.26&$-$1.23\\
&10&347&417&416&$-$16.59&0.24&322&386&389&$-$17.22&$-$0.77&297&358&361&$-$17.73&$-$0.83\\
      \midrule
    \multirow{2}{*}{24 mo.}
&3&84&118&120&$-$30.00&$-$1.67&83&116&119&$-$30.25&$-$2.52&79&111&114&$-$30.70&$-$2.63\\
&4&110&163&166&$-$33.73&$-$1.81&108&160&162&$-$33.33&$-$1.23&103&152&154&$-$33.12&$-$1.30\\
      \bottomrule
    \end{tabular}
    }
    \caption{MBALT 2017-A Corollary~\ref{cor:cap_APV} pricing results (APV)
    comparison to prospectus approach \citep{mercedes_2017} (Con.) and actual
    realized cash flows (Act.) including percentage differences for a pool of
    29{,}845 36-month lease contracts.  The results
    of Corollary~\ref{cor:cap_APV} generally fall well within the prospectus
    approach, especially over short pricing horizons and as the observation
    window increases.  
    The various discount interest rate curves 
    may be found in Appendix~\ref{sec:app_detail}.  All figures not in 
    percentages are in millions.}
    \label{tab:apv_results}
\end{sidewaystable}


\subsection{Quantifying estimator uncertainty}
\label{subsec:quant_est}

In this section, we will illustrate how to use Theorem~\ref{thm:haz_norm} to
quantify the estimator uncertainty of price point estimates made with 
Corollary~\ref{cor:cap_APV}.  Before proceeding, we first indicate that the
mutual independence assumption of $(X_i, Y_i)$ and $(X_j, Y_j)$ for 
$1 \leq i \neq j \leq n$ is likely satisfied for the MBALT 2017-A pool of 
leases (see Section~\ref{subsec:understand} as needed).  As a reminder, if
an investor believed such an assumption was not satisfied, the results of this
section may not be valid.

To obtain wider confidence intervals for illustrative purposes, we will
consider 24-month lease contracts, which make up a smaller portion of
MBALT 2017-A (866 leases out of a total of 50{,}402).  The smaller sample
is for illustrative purposes only; this analysis scales without issue.
As in Section~\ref{subsec:price}, we filtered the 866 
24-lease contracts for data irregularities, which left 493 
24-month lease contracts (see Appendix~\ref{sec:app_detail} as needed).

Figure~\ref{fig:APV_dist} presents the resulting 95\% confidence intervals for 
the hazard rate estimators using Theorem~\ref{thm:haz_norm} for these 493
24-month lease contracts (compare with the 36-month contracts in 
Figure~\ref{fig:m17_hazrate}).  The confidence intervals suggest that the 
hazard rate estimate random variables may reasonably fall
within such intervals, which will naturally flow through the calculations in
Corollary~\ref{cor:cap_APV}.  From the work in Theorem~\ref{thm:apv_varapv},
we can specify the exact distribution of the vector of hazard rate estimators,
$\hat{\bm{\Lambda}}_{\tau, n}$, and use simulation to assess the potential
variance of our APV estimates.  That is, suppose we have a
given observation window, say 6 months.  Our procedure is as follows:
\begin{enumerate}[label={[\arabic*]}]
\item Determine $\hat{\bm{\Lambda}}_{\tau, n}$ using \eqref{eq:est_haz};
\item Estimate $\bm{\Sigma}$ from Theorem~\ref{thm:haz_norm} using the
estimators $\hat{f}_{*, \tau, n}$ and $\hat{U}_{\tau,n}$ defined in
\eqref{eq:est_haz};
\item Use the delta-method \citep[][Theorem 5.3.5, pg. 261]{nitis_2000} to
find the log-adjusted multivariate normal distribution (to ensure the
confidence intervals for the hazard rates fall within 0 and 1);
\item For each of the desired number of replicates, simulate a realization of
hazard rates from the multivariate normal distribution in [3] and then
proceed to calculate $\text{APV}^{\kappa}_{\text{Trust}}$ from
Corollary~\ref{cor:cap_APV} with these simulated hazard rates;
\item Assess the potential estimation error though the distribution of stored
results from [4].
\end{enumerate}

We have done exactly this in Figure~\ref{fig:APV_dist} for an observation
window of 6 months, $\kappa = 3$, 4{,}000 replicates, for the sample of
493 24-month leases from the MBALT 2017-A transaction described above. It is
noteworthy the randomness of the hazard rate estimators can influence the
pricing calculation of Corollary~\ref{cor:cap_APV} by a range of over \$1M,
especially as the price was calculated to be approximately \$5.4M. We
also see the actual realization for these three months does fall within the
range of simulated estimates.  For this particular comparison, therefore,
we would say that the estimation process of Section~\ref{sec:est} in
combination with the cash flow model of Section~\ref{sec:cash_flow} was able
to correctly predict the future cash flows.

We close this section with a remark on interpreting the results of
Figure~\ref{fig:APV_dist}.
As the sample size increases, we would expect the variability of the estimator
to decrease by Theorem~\ref{thm:haz_norm}, which will lead to a more precise
point estimate of the price using Corollary~\ref{cor:cap_APV}.  This should not
be conflated with less risk inherent in the future cash flows, however.  The
variability from the cash flows within \eqref{eq:PV} is due to the randomness
of the lifetime random variable, $X$, and the delay and residual random
variables, $D_X$ and $Z_X$, respectively.  The volatility of these random
variables will depend on the distributional estimates produced by the
underlying data and not on the variability of the estimators themselves
(though the latter may be of interest, too).  Our process has produced an
estimation and pricing process that we can expect to be asymptotically
unbiased; it does not suggest that the cash flow risk will decline as $n$
grows.

\begin{figure}[tbh]
    \centering
    \includegraphics{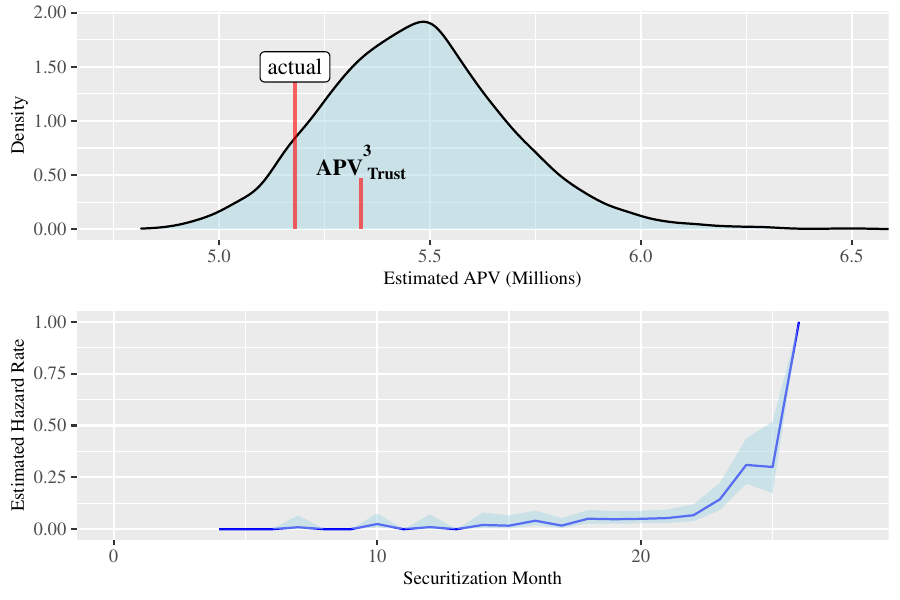}
    \caption{Illustrative example of Corollary~\ref{cor:cap_APV} 
    used to quantify estimator uncertainty.  Data represents 493 
    24-month leases from MBALT 2017-A with an observation window 
    of 6 months.  The
    hazard rates and 95\% confidence intervals fitted using 
    Theorem~\ref{thm:haz_norm} (see also Section~\ref{subsec:price}).
    The distribution of price point estimates was created with 
    4{,}000 replications.  The actual observed value and formulaic
    price $(\kappa = 3)$ are denoted within the figure.}
    \label{fig:APV_dist}
\end{figure}

\section{Conclusion}
\label{sec:conclusion}

Despite life insurers significant holdings in securitized assets including
ABS, it is difficult to find studies on this important fixed-income asset
class within the actuarial literature.  Further, current market pricing
techniques for these assets either rely on a non-random time-to-event model
or may not utilize the full asset-level disclosures of SEC Regulation AB II
\cite{reg_ab2}, which took effect in 2016. Our work fills this gap by
establishing an effective pricing process that makes use of discrete-time
survival analysis estimation techniques for incomplete data.

Broadly, there are two contributions of this article.  The first contribution
is the rigorous exposition of a framework capable of handling left-truncation
and right-censoring in discrete-time in Section~\ref{sec:prelim}.
This was necessary to derive the
asymptotic properties of the hazard rate estimators assuming a discrete $X$,
which we have done in Theorem~\ref{thm:haz_norm} (statement in
Section~\ref{sec:est} and proof in Appendix~\ref{sec:asym_proof}) for the
first time in the case of left-truncation and right-censoring in the
statistical literature.  Note that this is a generalization of the
results of \citet{lautier_2021}, which was only valid for the case of
left-truncation.  The second contribution is the pricing formula of
Theorem~\ref{thm:apv_varapv} (statement in Section~\ref{sec:cash_flow} and
proof in Appendix~\ref{subsec:thm_apv}), which is effectively an expected or
actuarial present value that relies on the lifetime random variable
distribution.  The two contributions come together in that the distribution of
the lifetime random variable may be estimated using the results of
Section~\ref{sec:est}, and Theorem~\ref{thm:haz_norm}, under certain settings,
may be used to assess the
potential uncertainty of the pricing point estimator in
Theorem~\ref{thm:apv_varapv}.  We also provide a discussion of when our model
framework --- in particular the key independence assumptions --- is appropriate
and inappropriate in Section~\ref{subsec:understand}.  It is our opinion that
the independence assumptions are reasonable in many realistic scenarios, but
not all, such as with securitization pools of subprime credits.

The theoretical results of this paper are applied to a subset of
29{,}845 leases from the Mercedes-Benz 2017-A auto lease ABS bond in
Section~\ref{sec:app}. Specifically in Section~\ref{sec:app}, we first
provide discussion about why this particular bond meets both of the important
independence assumptions of Section~\ref{subsec:understand}.  Next,
we found our formulaic pricing method, which accounts for the variability of
the lease lifetime distribution, was capable of outperforming the standard
modeling method within an auto lease ABS prospectus \citep{mercedes_2017},
which assumes a non-random lifetime distribution
(Section~\ref{subsec:price}, Table~\ref{tab:apv_results}).  Our illustrative
application closed with a demonstration of how to use the asymptotic results
of Theorem~\ref{thm:haz_norm} to assess the price point estimator uncertainty
of Theorem~\ref{thm:apv_varapv} (Section~\ref{subsec:quant_est}).

We recognize that other model formulations may attempt
to connect economic variables to credit modeling at the loan level
\citep{deng_2000}.  These
models typically find some association between consumer behavior and rational
market behavior, such as the connection between prepayment behavior and the
implicit ``in-the-money" level a borrower finds himself with respect to his
loan and home value.  While some option-based models do have explanatory power,
they are often not enough to fully explain the significant heterogeneity
exhibited by borrowers \citep{deng_2000}.
We have thus elected to use a data-based approach that models when lessees
decide to terminate their leases rather than attempt to explain why lessees
terminate their leases.  The ``when" is of paramount importance in a cash flow
pricing exercise, and we find success with this approach in the realistic
setting of Section~\ref{sec:app}.
Fundamentally, therefore, our model produces a random
timing and amount of cash flows at the individual lease contract level based on
distributional estimations from historical performance.  The most important of
these distributions, the random time-until-contract-termination, may be
estimated using the techniques of Section~\ref{sec:est}.

That said, we acknowledge that the ``why" of consumer lessee
behavior also has its merits, particularly in economically driven research.
We postulate that a future version of the estimator \eqref{eq:est_haz} may be
generalized to incorporate covariates, which could blend the ``when" and ``why"
into a single model.  Of particular interest may be the connection between
current market interest rates and consumer behavior, especially for models of
consumer loans or residential mortgages.  At present, we leave
this problem open to future research.

The results of this paper may be suitable for other forms of ABS besides auto
lease bonds.  One ideal generalization is agency MBS, which accounts for about
\$250 billion of life insurer assets \citep{mcmenamin_2013}.
We recommend a generalization of the
estimation framework of Sections~\ref{sec:prelim} and \ref{sec:est} to the case
of competing risks for securitized assets in which an investor would prefer to
treat prepayment and default separately.  Additional details may be found in
Section~\ref{subsec:understand}.  It is our
opinion the combined estimation and pricing framework of this paper may be of
use for insurance products, too, though we also leave this problem open for 
future research.

\bibliographystyle{mcap}
\bibliography{ime}

\appendix

\section{Proofs}
\label{sec:proofs}

We first prove Theorem~\ref{thm:haz_norm} and subsequently
prove Theorem~\ref{thm:apv_varapv}.

\subsection{Proof of Theorem~\ref{thm:haz_norm}}
\label{sec:asym_proof}

We first define some helpful notation:
\begin{align}
    u_{\tau}(k,k') &= \Pr(Y_i \leq k \leq \min(X_i, C_i),
    Y_i \leq k' \leq \min(X_i, C_i)) \nonumber\\
    &= \Pr(Y_i \leq \min(k,k'), \max(k,k') \leq X_i,
    \max(k,k') \leq C_i) \nonumber\\
    &= \Pr(Y \leq \min(k,k'), X \geq \max(k,k'), C \geq \max(k,k')
    \mid X \geq Y) \nonumber\\
    &= \Pr(Y \leq \min(k,k'), X \geq \max(k,k'), C \geq \max(k,k'),
    X \geq Y) / \Pr(X \geq Y) \nonumber\\
    &= \Pr(Y \leq \min(k,k'), X \geq \max(k,k'), C \geq \max(k,k'))
    / \alpha \nonumber\\
    &= \frac{1}{\alpha} \Pr(X \geq \max(k,k')) \Pr(Y \leq \min(k,k'),
    \max(k,k') \leq C). \label{eq:c_kk}
\end{align}
Notice $u_{\tau}(k,k') = u_{\tau}(k',k)$ and $u_{\tau}(k,k) = U_{\tau}(k)$.
Further,
\begin{align}
    r_{\tau}(k,k') &= \Pr(X_i = \max(k,k'), Y_i \leq \min(k,k'),
    X_i \leq C_i) \nonumber\\
    &= \Pr(X = \max(k,k'), Y \leq \min(k,k'), X \leq C \mid X \geq Y)
    \nonumber\\
    &= \Pr(X = \max(k,k'), Y \leq \min(k,k'), X \leq C, X \geq Y) /
    \Pr(X \geq Y) \nonumber\\
    &= \Pr(X = \max(k,k'), Y \leq \min(k,k'), C \geq \max(k,k')) /
    \alpha \nonumber\\
    &= \frac{1}{\alpha} \Pr(X = \max(k,k'))
    \Pr(Y \leq \min(k,k'), \max(k,k') \leq C). \label{eq:r_kk}
\end{align}
Notice $r_{\tau}(k,k') = r_{\tau}(k',k)$ and
$r_{\tau}(k,k) = f_{*,\tau}(k)$.  We first state a lemma, and the
proof of Theorem~\ref{thm:haz_norm} follows.

\begin{Lemma}[$\hat{\mathbf{U}}_{\tau,n}$ Asymptotic Properties]
Define $\hat{\mathbf{U}}_{\tau,n} =
\big(\hat{U}_{\tau,n}(\Delta+1), \ldots, \hat{U}_{\tau,n}(\xi) \big)^{\top}$,
where $\hat{U}_{\tau,n}$ follows from \eqref{eq:est_haz}. Then,
\begin{enumerate}[label=(\roman*)]
	\item \begin{equation*}
			\hat{\mathbf{U}}_{\tau,n}
			\overset{\mathcal{P}}{\longrightarrow}
			\mathbf{U}_{\tau},
			\text{ as } n \rightarrow \infty;
		  \end{equation*}
	\item \begin{equation*}
    			\sqrt{n} ( \hat{\mathbf{U}}_{\tau,n} - \mathbf{U}_{\tau} )
    			\overset{\mathcal{L}}{\longrightarrow}
    			N( \mathbf{0}, \mathbf{\Sigma}_u),
    			\text{ as } n \rightarrow \infty,
		  \end{equation*}
\end{enumerate}
where $\mathbf{U}_{\tau} = (U_{\tau}(\Delta + 1), \ldots, U_{\tau}( \xi ))^{\top}$
with $U_{\tau}$ as defined in \eqref{eq:c_tau} and $\bm{\Sigma}_u$ is
a covariance matrix $\lVert \sigma_{k',k} \rVert$ such that
\begin{equation*}
	\sigma_{k',k} = \begin{cases} U_{\tau}(k)[1 - U_{\tau}(k)], & k' = k\\
	u_{\tau}(k',k) - U_{\tau}(k')U_{\tau}(k), & k' \neq k
	\end{cases},
\end{equation*}
for $k', k = \Delta + 1, \ldots, \xi$.
\label{lem:C_hat}
\end{Lemma}

\begin{proof}
Statement \textit{(i)} follows from \textit{(ii)}, so it is left to show
\textit{(ii)}. Observe
\begin{equation*}
  \hat{\mathbf{U}}_{\tau,n}
  = \begin{bmatrix}
  \hat{U}_{\tau,n}(\Delta + 1) \\ \vdots \\ \hat{U}_{\tau,n}(\xi)
  \end{bmatrix}
  = \begin{bmatrix}
  	\displaystyle \frac{1}{n} \sum_{i=1}^{n}
    \mathbf{1}_{Y_i \leq \Delta + 1 \leq \min(X_i,C_i)} \\
    \vdots \\
    \displaystyle \frac{1}{n} \sum_{i=1}^{n}
    \mathbf{1}_{Y_i \leq \xi \leq \min(X_i,C_i)}
    \end{bmatrix}
  =\frac{1}{n} \sum_{i=1}^{n}
  \begin{bmatrix}
  	Y_{\tau,\Delta + 1 (i)} \\
    \vdots \\
    Y_{\tau, \xi (i)}
  \end{bmatrix},
\end{equation*}
where $Y_{\tau, k(i)}$, $\Delta + 1 \leq k \leq \xi$ are independent and
identically distributed Bernoulli random variables with probability of
success given by
\begin{equation*}
\Pr(Y_i \leq k \leq \min(X_i,C_i)) =
\Pr(Y \leq k \leq \min(X,C) \mid X \geq Y) = U_{\tau}(k),
\end{equation*}
for $k = \Delta + 1, \ldots, \xi$. Thus, $E[Y_{\tau, k(i)}] = U_{\tau}(k)$ and
$\text{Var}[Y_{\tau, k(i)}] = U_{\tau}(k) (1 - U_{\tau}(k))$.  Now, since
\begin{equation*}
  \mathbf{1}_{Y_i \leq k' \leq \min(X_i, C_i)}
   \mathbf{1}_{Y_i \leq k \leq \min(X_i,C_i)} =
   \mathbf{1}_{Y_i \leq \min(k',k), X_i \geq \max(k',k), C_i \geq \max(k',k)},
\end{equation*}
we have
\begin{equation}
  E[Y_{\tau, k'(i)} Y_{\tau, k(i)}] =
  E[\mathbf{1}_{Y_i \leq \min(k',k), X_i \geq \max(k',k), C_i \geq \max(k',k)}]
  = u_{\tau}(k',k), \label{eq:EY2}
\end{equation}
for $k', k = \Delta + 1, \ldots, \xi$.  Thus,
\begin{align*}
    \text{Cov}[Y_{\tau, k'(i)} Y_{\tau, k(i)} ]
    &= E[Y_{\tau, k'(i)} Y_{\tau, k(i)}] -
    E[Y_{\tau, k'(i)}]E[Y_{\tau, k(i)}]\\
    &= u_{\tau}(k',k) - U_{\tau}(k')U_{\tau}(k).
\end{align*}
Recall that~\eqref{eq:EY2} reduces to $U_{\tau}(k)$ when $k' = k$.  Use the
multivariate Central Limit Theorem \citep[][Theorem 8.21, pg. 61]{lehmann_1998}
to complete the proof.
\end{proof}

We now prove Theorem~\ref{thm:haz_norm}.

\begin{proof}
Statement \textit{(i)} follows from \textit{(ii)}, so it is left to show
\textit{(ii)}.  Let $\Delta + 1 \leq k \leq \xi$ and observe
$\mathbf{1}_{X_i \leq C_i} \mathbf{1}_{\min(X_i, C_i) = k} =
\mathbf{1}_{X_i = k, X_i \leq C_i}$ and so
\begin{align*}
    \hat{\lambda}_{\tau,n}(k) - \lambda_{\tau}(k) &=
    \frac{ \frac{1}{n} \sum_{i=1}^{n} \mathbf{1}_{X_i = k, X_i \leq C_i}}
    {\hat{U}_{\tau, n}(k)} - \frac{ f_{*, \tau}(k) }{ U_{\tau}(k) }\\
    &= \frac{ \frac{1}{n} \sum_{i=1}^{n} \mathbf{1}_{X_i = k, X_i \leq C_i}
    U_{\tau}(k) -
    f_{*, \tau}(k)\hat{U}_{\tau,n}(k) }{ \hat{U}_{\tau,n}(k)U_{\tau}(k)}\\
    &= \bigg[ \frac{1}{\hat{U}_{\tau,n}(k)U_{\tau}(k)} \bigg] \frac{1}{n}
    \sum_{i=1}^{n} \{ \mathbf{1}_{X_i=k,X_i \leq C_i} U_{\tau}(k) -
    f_{*,\tau}(k) \mathbf{1}_{Y_i \leq k \leq \min(X_i,C_i)} \}.
\end{align*}
Further define
\begin{equation*}
    Z_{\tau, k(i)} = \mathbf{1}_{X_i=k,X_i \leq C_i} U_{\tau}(k) -
    f_{*,\tau}(k) \mathbf{1}_{Y_i \leq k \leq \min(X_i,C_i)}.
\end{equation*}
Hence,
\begin{equation*}
\hat{\bm{\Lambda}}_{\tau,n} - \bm{\Lambda}_{\tau} =
\mathbf{A}_{\tau,n}  \frac{1}{n} \sum_{i=1}^{n}
\begin{bmatrix}
Z_{\tau, \Delta + 1 (i)} \\ \vdots \\ Z_{\tau, \xi(i)}
\end{bmatrix},
\end{equation*}
where
$\mathbf{A}_{\tau,n} =
\text{diag}([\hat{U}_{\tau,n}(\Delta+1)U_{\tau}(\Delta+1)]^{-1},
\ldots, [\hat{U}_{\tau,n}(\xi)U_{\tau}(\xi)]^{-1})$.  That is,
\begin{equation*}
\hat{\bm{\Lambda}}_n - \bm{\Lambda} = \mathbf{A}_{\tau,n}  \frac{1}{n}
\sum_{i=1}^{n} \mathbf{Z}_{\tau,(i)},
\end{equation*}
where $\mathbf{Z}_{\tau,(i)} =
(Z_{\tau,\Delta+1(i)}, \ldots, Z_{\tau,\xi(i)})^{\top}$,
$1 \leq i \leq n$ are independent and identically distributed random vectors.
We will also subsequently show that the components of random vector
$\mathbf{Z}_{\tau,(i)}$ are uncorrelated.

First notice $\mathbf{1}_{X_i = k, X_i \leq C_i}$ is a Bernoulli random
variable with probability parameter $f_{*, \tau}(k)$.  Similarly,
$\mathbf{1}_{Y_i \leq k \leq \min(X_i, C_i)}$ is a Bernoulli random
variable with probability parameter $U_{\tau}(x)$. Thus,
\begin{align*}
    E[Z_{\tau, k(i)}] &= E[ \mathbf{1}_{X_i=k,X_i \leq C_i}] U_{\tau}(k) -
    f_{*,\tau}(k) E[ \mathbf{1}_{Y_i \leq k \leq \min(X_i,C_i)}]\\
    &= f_{*, \tau}(k) U_{\tau}(k) - f_{*,\tau}(k) U_{\tau}(k)\\
    &= 0.
\end{align*}
We now show
\begin{equation}
    \text{Cov}[Z_{k(i)}, Z_{k'(i)}] = \begin{cases}
    U_{\tau}(k) f_{*, \tau}(k) [U_{\tau}(k) - f_{*, \tau}(k)], &
    k = k'\\
    0, & k \neq k'. \end{cases}
    \label{eq:cov_Z}
\end{equation}
Since $E[Z_{\tau, k(i)}]=0$, we have
\begin{align*}
    \text{Cov}[Z_{\tau,k(i)}, Z_{\tau,k'(i)}] ={}& E[Z_{\tau,k(i)}Z_{\tau,k'(i)}]\\
    ={}& E \bigg[ \bigg( \mathbf{1}_{X_i=k,X_i \leq C_i}
    U_{\tau}(k) - f_{*,\tau}(k) \mathbf{1}_{Y_i \leq k \leq \min(X_i,C_i)} \bigg)\\
    & \bigg( \mathbf{1}_{X_i=k',X_i \leq C_i} U_{\tau}(k') -
    f_{*,\tau}(k') \mathbf{1}_{Y_i \leq k' \leq \min(X_i,C_i)} \bigg) \bigg]\\
    ={}& U_{\tau}(k)U_{\tau}(k') E[ \mathbf{1}_{X_i=k,X_i \leq C_i}
    \mathbf{1}_{X_i=k',X_i \leq C_i} ]\\
    &- f_{*, \tau}(k)U_{\tau}(k') E[ \mathbf{1}_{Y_i \leq k \leq \min(X_i,C_i)}
    \mathbf{1}_{X_i=k',X_i \leq C_i}]\\
    &- f_{*, \tau}(k')U_{\tau}(k) E[ \mathbf{1}_{Y_i \leq k' \leq \min(X_i,C_i)}
    \mathbf{1}_{X_i=k,X_i \leq C_i}]\\
    &+ f_{*,\tau}(k)f_{*,\tau}(k')
    E[ \mathbf{1}_{Y_i \leq k \leq \min(X_i,C_i)}
    \mathbf{1}_{Y_i \leq k' \leq \min(X_i,C_i)}]
\end{align*}
We proceed by cases.

Case 1: $k = k'$.

Since $\mathbf{1}_{X_i=k,X_i \leq C_i}\mathbf{1}_{X_i=k',X_i \leq C_i} =
\mathbf{1}_{X_i=k,X_i \leq C_i}$, $E[\mathbf{1}_{X_i=k,X_i \leq C_i}
\mathbf{1}_{X_i=k',X_i \leq C_i}] = f_{*, \tau}(k)$.  Additionally,
\begin{align*}
    \mathbf{1}_{Y_i \leq k \leq \min(X_i,C_i)}\mathbf{1}_{X_i=k',X_i \leq C_i}
    &= \mathbf{1}_{Y_i \leq k' \leq \min(X_i,C_i)}
    \mathbf{1}_{X_i=k,X_i \leq C_i}\\
    &= \mathbf{1}_{Y_i \leq k \leq \min(X_i, C_i), X_i = k, X_i \leq C_i}\\
    &= \mathbf{1}_{Y_i \leq X_i \leq X_i, X_i = k, X_i \leq C_i}\\
    &= \mathbf{1}_{X_i = k, X_i \leq C_i}.
\end{align*}
Therefore,
\begin{equation*}
    E[ \mathbf{1}_{Y_i \leq k \leq \min(X_i,C_i)}
    \mathbf{1}_{X_i=k',X_i \leq C_i}] =
    E[ \mathbf{1}_{Y_i \leq k' \leq \min(X_i,C_i)}
    \mathbf{1}_{X_i=k,X_i \leq C_i}] = f_{*, \tau}(k).
\end{equation*}
Finally,
\begin{equation*}
    \mathbf{1}_{Y_i \leq k \leq \min(X_i,C_i)}
    \mathbf{1}_{Y_i \leq k' \leq \min(X_i,C_i)} =
    \mathbf{1}_{Y_i \leq k \leq \min(X_i,C_i)}.
\end{equation*}
Thus, $E[ \mathbf{1}_{Y_i \leq k \leq \min(X_i,C_i)}
\mathbf{1}_{Y_i \leq k' \leq \min(X_i,C_i)}] = U_{\tau}(k)$.  Replace these
expectations in $E[Z_{k(i)} Z_{k'(i)}]$ to write
\begin{equation*}
    \text{Cov}[Z_{k(i)}, Z_{k'(i)}] = U_{\tau}(k)f_{*,\tau}(k)[ U_{\tau}(k) -
    f_{*, \tau}(k)].
\end{equation*}

Case 2: $k \neq k'$.

Since $\mathbf{1}_{X_i=k,X_i \leq C_i}\mathbf{1}_{X_i=k',X_i \leq C_i} = 0$,
$E[\mathbf{1}_{X_i=k,X_i \leq C_i}\mathbf{1}_{X_i=k',X_i \leq C_i}]=0$.  Assume
$k < k'$.  Then
\begin{align*}
    \mathbf{1}_{Y_i \leq k \leq \min(X_i, C_i)} \mathbf{1}_{X_i=k',X_i\leq C_i}
    &= \mathbf{1}_{Y_i \leq k \leq \min(X_i,C_i), X_i = k', X_i \leq C_i}\\
    &= \mathbf{1}_{Y_i \leq k, X_i = k', X_i \leq C_i}.
\end{align*}
Therefore, $E[\mathbf{1}_{Y_i \leq k \leq \min(X_i, C_i)}
\mathbf{1}_{X_i=k',X_i\leq C_i} ] = \Pr(X_i = k', Y_i \leq k,
X_i \leq C_i)$.  Further, when $k < k'$
\begin{equation*}
    \mathbf{1}_{Y_i \leq k' \leq \min(X_i, C_i)} \mathbf{1}_{X_i=k,X_i\leq C_i}
    = \mathbf{1}_{Y_i \leq k' \leq \min(X_i, C_i), X_i = k, X_i \leq C_i} = 0.
\end{equation*}
Thus, $E[\mathbf{1}_{Y_i \leq k' \leq \min(X_i, C_i)}
\mathbf{1}_{X_i=k,X_i\leq C_i}] = 0$.  Now, if instead $k > k'$, then by symmetry,
\begin{equation*}
E[\mathbf{1}_{Y_i \leq k' \leq \min(X_i, C_i)} \mathbf{1}_{X_i=k,X_i\leq C_i}] =
\Pr(X_i = k, Y_i \leq k', X_i \leq C_i),
\end{equation*}
and $E[ \mathbf{1}_{Y_i \leq k \leq \min(X_i, C_i)}
\mathbf{1}_{X_i=k',X_i\leq C_i}] = 0$.  Thus, we can generalize and claim
\begin{align*}
    & f_{*, \tau}(k)U_{\tau}(k')E[ \mathbf{1}_{Y_i \leq k \leq \min(X_i, C_i)} \mathbf{1}_{X_i=k',X_i\leq C_i} ]
    + f_{*, \tau}(k')U_{\tau}(k)E[ \mathbf{1}_{Y_i \leq k \leq \min(X_i, C_i)}
    \mathbf{1}_{X_i=k',X_i\leq C_i}]\\
    =& f_{*, \tau}(\min(k,k')) U_{\tau}(\max(k,k')) r_{\tau}(k,k').
\end{align*}
Lastly, notice $E[ \mathbf{1}_{Y_i \leq k \leq \min(X_i,C_i)}
\mathbf{1}_{Y_i \leq k' \leq \min(X_i,C_i)}] = u_{\tau}(k,k')$.  Replace these
expectations in $E[Z_{k(i)}Z_{k'(i)}]$ to write
\begin{align*}
    \text{Cov}(Z_{k(i)}, Z_{k'(i)})
    ={}& -f_{*,\tau}(\min(k,k')) U_{\tau}(\max(k,k')) r_{\tau}(k,k')\\
    &+ f_{*,\tau}(\min(k,k')) f_{*,\tau}(\max(k,k')) u_{\tau}(k,k')\\
    ={}& f_{*,\tau}(\min(k,k')) \{ f_{*,\tau}(\max(k,k')) u_{\tau}(k,k') -
    r_{\tau}(k,k') U_{\tau}(\max(k,k')) \}.
\end{align*}
However, using \eqref{eq:c_kk} and \eqref{eq:r_kk},
\begin{align*}
    f_{*,\tau}(\max(k,k')) u_{\tau}(k,k') &=
    \bigg[ \frac{ \Pr(X = \max(k,k')) \Pr(Y \leq \max(k,k') \leq C) }
    { \alpha } \bigg]\\
    &\times \Bigg[\frac{ \Pr(X \geq \max(k,k')) \Pr(Y \leq \max(k,k'),
    \max(k,k') \leq C)}
    {\alpha} \bigg]\\
    &= \bigg[ \frac{ \Pr(X=\max(k,k')) \Pr(Y \leq \min(k,k'),
    \max(k,k') \leq C)}{\alpha} \bigg]\\
    & \times \bigg[ \frac{ \Pr(X \geq \max(k,k')) \Pr(Y \leq \max(k,k')
    \leq C)}{\alpha} \bigg]\\
    &= r_{\tau}(k,k') U_{\tau}(\max(k,k')).
\end{align*}
Thus, $\text{Cov}[Z_{k(i)}, Z_{k'(i)}] = 0$ when $k \neq k'$.  This confirms
\eqref{eq:cov_Z}.  Now define
\begin{equation*}
\mathbf{D}_{\tau} = \text{diag}\big(
U_{\tau}(\Delta+1)f_{*,\tau}(\Delta+1)
[U_{\tau}(\Delta+1) - f_{*,\tau}(\Delta+1)], \ldots,
U_{\tau}(\xi)f_{*,\tau}(\xi)
[U_{\tau}(\xi) - f_{*,\tau}(\xi)] \big),
\end{equation*}
and
\begin{equation*}
\bar{\mathbf{Z}}_{\tau,n} =
\frac{1}{n} \sum_{i=1}^{n} \mathbf{Z}_{\tau,(i)}.
\end{equation*}
Thus, by the multivariate Central Limit Theorem
\citep[][Theorem 8.21, pg. 61]{lehmann_1998} we may claim
\begin{equation*}
\sqrt{n}(\bar{\mathbf{Z}}_{\tau,n} - \bm{0})
\overset{\mathcal{L}}{\longrightarrow} N(\bm{0},\mathbf{D}_{\tau}),
\text{ as } n \rightarrow \infty.
\end{equation*}
Next define
\begin{equation*}
\mathbf{V}_{\tau} = \text{diag} \big(
U_{\tau}(\Delta+1)^{-2}, \ldots, U_{\tau}(\xi)^{-2}),
\end{equation*}
and use Lemma~\ref{lem:C_hat} to claim
$\mathbf{A}_{\tau,n} \overset{\mathcal{P}}{\longrightarrow} \mathbf{V}_{\tau}$,
as $n \rightarrow \infty$.  Therefore, by multivariate Slutsky's Theorem
\citep[][Theorem 5.1.6, pg. 283]{lehmann_1998},
\begin{equation*}
\sqrt{n} \big( \mathbf{A}_{\tau,n} \bar{\mathbf{Z}}_{\tau,n} \big)
\overset{\mathcal{L}}{\longrightarrow}
N(\bm{0}, \mathbf{V}_{\tau} \mathbf{D}_{\tau} \mathbf{V}_{\tau}^{\top})
\text{ as } n \rightarrow \infty.
\end{equation*}
Finally, observe
$\mathbf{V}_{\tau} \mathbf{D}_{\tau} \mathbf{V}_{\tau}^{\top} = \bm{\Sigma}$
and $\mathbf{A}_{\tau,n} \bar{\mathbf{Z}}_{\tau,n} =
\hat{\bm{\Lambda}}_{\tau,n} - \bm{\Lambda}_{\tau}$ to complete the
proof.
\end{proof}

\subsection{Proof of Theorem~\ref{thm:apv_varapv}}
\label{subsec:thm_apv}

\begin{proof}
By repeated use of law of total expectation
\cite[][Theorem 3.3.1, pg. 112]{nitis_2000} and \eqref{eq:PV},
\begin{equation}
    \text{APV}_i = \mathbf{E} [\text{PV}_i]
    = \mathbf{E}_X \{
    \mathbf{E}_D [
    \mathbf{E}_Z( \text{PV}_i \mid X_i, D_{X_i})
    \mid X_i] \}.
    \label{eq:APV_i}
\end{equation}
Now,
\begin{equation*}
    \mathbf{E}_Z( \text{PV}_i \mid X_i, D_{X_i})
    = W_i(X_i, D_{X_i}) + R_i(X_i) \mathbf{E}(Z \mid X = X_i).
\end{equation*}
Thus,
\begin{align*}
    \mathbf{E}_D [ \mathbf{E}_Z( \text{PV}_i \mid X_i, D_{X_i}) \mid X_i]
    &= \mathbf{E}_D
    \bigg[ W_i(X_i, D_{X_i}) + R_i(X_i) \mathbf{E}(Z \mid X = X_i)
    \bigg{|} X_i \bigg]\\
    &= \mathbf{E}_D [W_i(X_i, D_{X_i}) \mid X_i]
    + R_i(X_i) \mathbf{E}(Z \mid X = X_i)\\
    &= \bigg{ \{ }
    \sum_{k=0}^{\varphi} W_i(X_i, k) \Pr^*(D_{m} = k)
    \bigg{ \} }
    + R_i(X_i) \mathbf{E}(Z \mid X = X_i).
\end{align*}
Hence, returning to \eqref{eq:APV_i}
\begin{align*}
    \text{APV}_i &=
    \mathbf{E}_X \{ \mathbf{E}_D
    [\mathbf{E}_Z( \text{PV}_i \mid X_i, D_{X_i}) \mid X_i] \}\\
    &=
    \mathbf{E}_X \bigg[
    \bigg{ \{ }
    \sum_{k=0}^{\varphi} W_i(X_i, k) \Pr^*(D_{m} = k)
    \bigg{ \} }
    + R_i(X_i) \mathbf{E}(Z \mid X = X_i)
    \bigg]\\
    &=
    \sum_{m=x_{\varepsilon(i)}}^{\xi}
    \bigg(
    \bigg{ \{ }
    \sum_{k = 0}^{\varphi }
    W_i(m,k) \Pr^*(D_{m} = k)
    \bigg{ \} }
    + R_i(m) \mathbf{E}(Z \mid X = m)
    \bigg) p_{x_{\varepsilon(i)}}^{m}.
\end{align*}
The proof is complete by the linear property of expectations
\citep[][Theorem 3.3.2, pg. 116]{nitis_2000}.
\end{proof}

\section{Extended application details}
\label{sec:app_detail}

The following information pertains to Section~\ref{sec:app}.

\subsubsection*{Observed data irregularities}

Per the opening discussion of Section~\ref{subsec:price}, some records were
removed due to difficulties interpreting the reported cash flows.  In this
section, we present examples of two such records in
Table~\ref{tab:bad_samp_life}.

For example asset number~3, 
we can see that there are two positive
payments in the residual column (\texttt{liquidationProceedsAmount}). While
it appears the smaller payment of 3{,}728 at age 37 represents approximately
nine monthly payments of 417, it is not clearly indicated in the data.
Modeling such a lease contract will require making assumptions about
the payment of 3{,}728 that we cannot verify.  Instead, as the model of
Section~\ref{sec:cash_flow} is a new proposal, we have attempted to remove the
potential confounding effects of data irregularities and focused on
demonstrating the model is capable of pricing cash flows that follow a clear
pattern, such as those in Table~\ref{tab:samp_life}.

For example asset number~4, 
we see two large residual payments at ages
32 and 40.  Further, we see that the termination indicator
(\texttt{terminationIndicator}) occurs at age 35, despite the large residual
payment occurring later, at age 40.  This cash flow pattern
is very difficult to translate accurately from the reported data into a
lease contract outcome without more information.

\begin{center}
\begin{table}[tbh]
    \centering
    \begin{tabular}{ccccccccc}
    & \multicolumn{4}{c}{Ex. Asset Num: 3} & 
    \multicolumn{4}{c}{Ex. Asset Num: 4}\\ 
    \cmidrule(lr){1-1} \cmidrule(lr){2-5} \cmidrule(lr){6-9}
    Obs. Month & Age & Pmt & Resid. & Term. Ind. &
    Age & Pmt & Resid. & Term. Ind.\\
    \cmidrule(lr){1-1} \cmidrule(lr){2-5} \cmidrule(lr){6-9}
    1 & 30 & 834 & 0 & -- & 30 & 4{,}315 & 0 & --\\
    2 & 31 & 417 & 0 & -- & 31 & 0 & 0 & --\\
    3 & 32 & 417 & 0 & -- & 32 & 2{,}157 & 20{,}194 & --\\
    4 & 33 & 417 & 0 & -- & 33 & 0 & 0 & --\\
    5 & 34 & 417 & 0 & -- & 34 & 0 & 0 & --\\
    6 & 35 & 0 & 0 & -- & 35 & 0 & 0 & 1\\
    7 & 36 & 807 & 0 & -- & 36 & 0 & 0 & --\\
    8 & 37 & 416 & 3{,}728 & -- & 37 & 0 & 0 & -- \\
    9 & 38 & 0 & 0 & -- & 38 & 0 & 0 & --\\
    10 & 39 & 0 & 14{,}974 & 1 & 39 & 0 & 0 & --\\
    11 & -- & -- & -- & -- & 40 & 0 & 36{,}934 & --\\
    \cmidrule(lr){1-1} \cmidrule(lr){2-5} \cmidrule(lr){6-9}
    \end{tabular}
    \caption{MBALT 2017-A sample lives of cash flow data irregularities}
    \label{tab:bad_samp_life}
\end{table}
\end{center}

\subsubsection*{Interest rate scenarios}

Figure~\ref{fig:spot_rate} presents the two non-zero interest rate scenarios
for the results of Section~\ref{sec:app}.

\begin{figure}[tbh]
    \centering
    \includegraphics{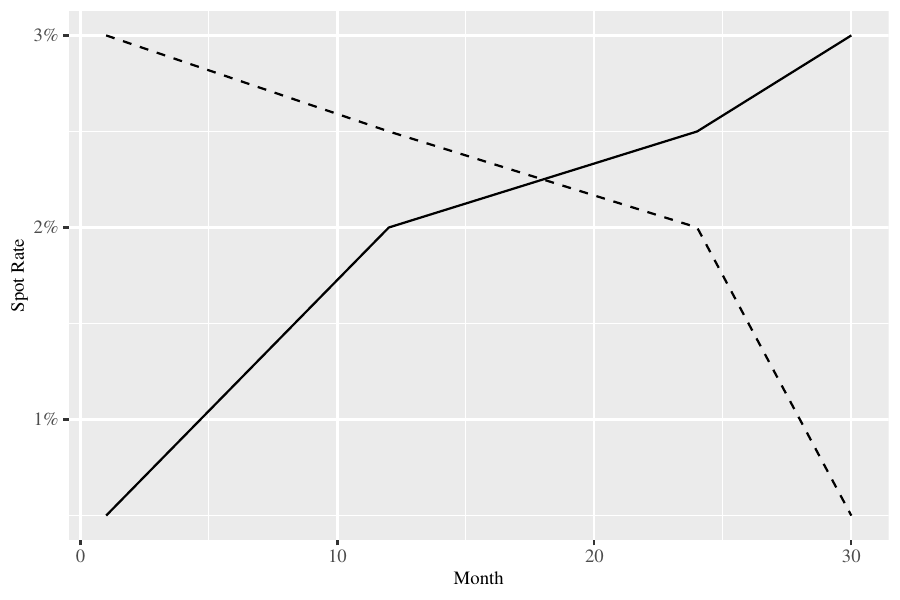}
    \caption{Spot rates for Section~\ref{sec:app}; standard (solid),
    inverted (dashed)}
    \label{fig:spot_rate}
\end{figure}

\end{document}